\documentclass[11pt]{article}
\usepackage[margin=1in]{geometry}
\usepackage{amsmath,amssymb,amsthm,thmtools}
\usepackage[dvipsnames]{xcolor}
\usepackage{xifthen}
\usepackage{bbm}
\declaretheorem[name=Theorem, numberwithin=section]{theorem}

\declaretheorem[name=Lemma,sibling=theorem]{lemma}
\declaretheorem[name=Corollary, sibling=theorem]{corollary}

\theoremstyle{definition}
\declaretheorem[name=Definition, sibling=theorem, style=definition]{definition}

\declaretheorem[name=Remark,sibling=theorem]{remark}

\usepackage[numbers,sort]{natbib}
\usepackage[colorlinks=true]{hyperref}
\hypersetup{
    colorlinks=true,
    linkcolor=PineGreen,
    filecolor=magenta,      
    urlcolor=cyan,
    citecolor=WildStrawberry,
}
\usepackage[capitalise,nameinlink,noabbrev]{cleveref}

\crefname{appendix}{appendix}{appendices}
\Crefname{appendix}{Appendix}{Appendices}

\newcommand{\EMD}{\mathsf{EMD}}
\newcommand{\TV}{\mathsf{TV}}
\newcommand{\OPT}{\mathsf{OPT}}
\newcommand{\MaxCSP}{\mathsf{MaxCSP}}
\newcommand{\set}[1]{\left\{#1\right\}}
\newcommand{\abs}[1]{\left|#1\right|}
\newcommand{\1}[1]{\mathbbm 1\left[#1\right]}

\newcommand{\argmin}{\mathop{\arg\min}}
\newcommand{\inner}[2]{\langle #1,#2\rangle}
\newcommand{\E}[2][]{ \ifthenelse{\isempty{#1}}
  {\mathbb{E}\left[#2\right]}
  {\mathop{\mathbb{E}}_{#1}\left[#2\right]} }
\renewcommand{\Pr}[2][]{ \ifthenelse{\isempty{#1}}
  {\mathbf{Pr}\left[#2\right]} {\mathop{\mathbf{Pr}}_{#1}\left[#2\right]} }

\title{Sensitivity Lower Bounds via Locally Testable Codes}
\author{Yuichi Yoshida\thanks{Supported by JSPS KAKENHI Grant Number 24K02903.}\\ National Institute of Informatics \\ \texttt{yyoshida@nii.ac.jp}
\and 
Zihan Zhang\thanks{Supported by JST SPRING, Japan Grant Number JPMJSP2104.} \\ National Institute of Informatics \\ \texttt{zihan@nii.ac.jp}}
\date{}

\begin{document}

\maketitle
\begin{abstract}
Sensitivity quantifies how far an algorithm's output can move in Hamming distance when a single input element is perturbed.
In this work, we present a general scheme that turns any locally testable code (LTC) into a sensitivity lower bound for the constraint satisfaction problem encoded by the code.

Instantiating the scheme with the $c^3$-LTC yields a constant $\varepsilon > 0$ for which every $(1-\varepsilon)$-approximation algorithm for satisfiable Max E3LIN2 must have sensitivity $\Omega(n)$, strengthening the previous $\Omega(n^\delta)$ lower bound that held only for general instances. 
Standard reductions then give an $n^{1-O(\varepsilon)}$ sensitivity lower bound for $n^{-\varepsilon}$-approximation algorithms for the maximum clique problem, and an $\Omega(k)$ sensitivity bound for $(1-\varepsilon)$-approximation algorithms for the maximum $k$-coverage problem, among others.
These lower bounds match the trivial upper bounds and imply that even slightly stable algorithms cannot achieve these approximations.

A second instantiation, using a simple repetition code, shows that any $(1-\varepsilon)$-approximation algorithm for Max Cut on bipartite graphs must have sensitivity $\Omega(1/\varepsilon)$ as long as $\varepsilon = O(1/\sqrt{n})$, extending the prior result for exact algorithms (the regime $\varepsilon < 1/n$). Thus even on bipartite graphs, where a perfect cut always exists, near-optimal solutions cannot be maintained stably.

Our sensitivity lower bounds also have two further applications.
First, our argument yields sensitivity lower bounds in an averaged setting, implying that any nearly optimal randomized algorithm for Max E3SAT has linearly many output bits that, after fixing the random seed, are Boolean functions with large influence.
Second, via the sensitivity-to-locality connection, our worst-case sensitivity lower bounds imply locality lower bounds in the non-signaling model, which generalizes $\mathsf{LOCAL}$ and quantum-$\mathsf{LOCAL}$.

\end{abstract}

\thispagestyle{empty}
\newpage
\thispagestyle{empty}
\tableofcontents
\setcounter{page}{0}
\newpage

\section{Introduction}\label{sec:intro}

The \emph{(worst-case) sensitivity} of an algorithm is the maximum change in its output, measured in Hamming distance, when a single input element is perturbed.
For randomized algorithms we compare output distributions using earth mover's distance (EMD).
Sensitivity offers a foundational robustness metric for algorithms and has deep connections to sublinear-time and distributed computation~\cite{fleming2026sensitivity}.
Sensitivity also underpins consistent decision making and reproducible knowledge discovery. 
Prior work has developed low-sensitivity algorithms for numerous combinatorial optimization problems~\cite{varma2023average,yoshida2021sensitivity,yoshida2026low,kumabe2022average,kumabe2022average-knapsack}, clustering tasks~\cite{yoshida2022average,peng2020average,li2025average}, and decision tree learning~\cite{hara2023average}.

On the lower-bound side, it is known that computing $2$-coloring of a bipartite graph requires sensitivity $\Omega(n)$~\cite{varma2023average}.
Recently, Fleming and Yoshida~\cite{fleming2026sensitivity} leveraged probabilistically checkable proofs (PCPs) to show that, for canonical constraint satisfaction problems (CSPs) such as Max E3SAT, Max E3LIN2, and Max Cut, there exist constants $\varepsilon, \delta > 0$ for which any $(1-\varepsilon)$-approximation algorithm must have sensitivity $\Omega(n^\delta)$, where $n$ denotes the number of variables or vertices.
Standard reductions then yield $\Omega(n^\delta)$ lower bounds for approximation algorithms for the maximum clique and minimum vertex cover problems. 
While their framework provides the first nontrivial lower bounds for approximation algorithms, it nevertheless has several shortcomings:
\begin{itemize}
\item The $\Omega(n^\delta)$ sensitivity bound is evidently far from tight.
\item The lower bounds for Max E3LIN2 and Max Cut are obtained via gadget reductions from Max E3SAT, so they say nothing about approximation algorithms for satisfiable instances.
\item The argument applies only to worst-case sensitivity, leaving the behavior under randomly chosen perturbations unresolved.
\end{itemize}
In this work we address these shortcomings by relating sensitivity to locally testable codes (LTCs) and showing how the existence of an LTC yields sensitivity lower bounds.
We outline the framework below.

\subsection{Our Contributions}

\paragraph{Sensitivity Lower Bounds for CSPs via LTCs.}

\emph{Locally testable codes (LTCs)}~\cite{goldreich2006locally} are error-correcting codes whose membership can be checked by probing only a constant number of coordinates of the received word. 
An LTC is called a \emph{strong LTC}~\cite{goldreich2019strong} if there exists $\kappa > 0$ such that the rejection probability is at least $\kappa$ times the relative distance from the code.
The parameter $\kappa$ is called the \emph{detection probability} of the code.
As we only consider strong LTCs in this work, we also call them LTCs.
In this work, we mainly consider linear codes, meaning they form a subspace of $\mathbb F_p^n$.

An LTC naturally induces a CSP instance: we regard each coordinate of the input word as a variable, and each test performed by the tester as a constraint on those variables.
The \emph{variable degree} $\theta$ of the tester is the maximum number of constraints that query any single coordinate.
If the input is a codeword the instance is satisfiable, whereas inputs far from the code violate a fraction of the constraints that is proportional to their distance from the code.
For an LTC $C$, let $\mathsf{MaxCSP}(C)$ denote the Max CSP whose constraints are restricted to those used in the tester for $C$.
Throughout the introduction, we measure sensitivity for CSP instances with respect to deleting a single constraint.
The main contribution of this work is showing a general scheme that converts a linear LTC $C$ to a sensitivity lower bound of $\mathsf{MaxCSP}(C)$.
\begin{theorem}\label{thm:main-intro}
    Let $\varepsilon > 0$ and $C$ be a linear LTC over $\mathbb F_p$ of block length $n$, relative distance $\delta$, rate $r$, variable degree $\theta$, with detection probability $\kappa$.
    Then, any $(1-\varepsilon)$-approximation algorithm for $\mathsf{MaxCSP}(C)$ has sensitivity at least
    \[
        \frac{n}{2\theta}\left(\delta\left(1-\frac{p}{p-1}\frac{\varepsilon}{r\kappa}\right)-\frac{2\varepsilon}{\kappa}\right) - \frac{1}{2\theta}.
    \]
    even on satisfiable instances.
\end{theorem}
We emphasize that this lower bound is information-theoretic and does not rely on any complexity assumption.

We first instantiate \Cref{thm:main-intro} using the $c^3$-LTC of~\cite{dinur2022locally,panteleev2022asymptotically}, which enjoys constant relative distance, constant rate, a constant number of queries, and constant detection probability.
The associated CSP is Max E3LIN2, whose constraints are linear equations of the form $x + y + z = b$ over $\mathbb F_2$.
\begin{theorem}\label{thm:e3lin2-intro}
    There exists $\varepsilon > 0$ such that every $(1-\varepsilon)$-approximation algorithm for Max E3LIN2 must have sensitivity $\Omega(n)$ even when the input instances are satisfiable.
\end{theorem}
Since the assignment length is $n$, every algorithm automatically has sensitivity at most $n$, so \Cref{thm:e3lin2-intro} matches this trivial upper bound up to the hidden constant.
Equivalently, for satisfiable Max E3LIN2, even slightly stable algorithms cannot achieve a $(1-\varepsilon)$-approximation.

Previously, the best lower bound was $\Omega(n^\delta)$ for some $\delta > 0$, and it applied only to arbitrary (not necessarily satisfiable) instances~\cite{fleming2026sensitivity}.
This non-tight exponent comes from the PCP-based reduction: it blows an $n$-variable instance up to size about $N := n^{1/\delta}$, so even if the $\Omega(n)$ sensitivity lower bound survives the reduction, expressing the result in terms of the enlarged instance size $N$ leaves only an $\Omega(N^\delta)$ lower bound.

Via standard reductions from Max E3LIN2 to Max 3SAT and Max Cut (see, e.g.,~\cite{trevisan2000gadgets}), we can show that $(1-\varepsilon)$-approximation algorithms for Max E3SAT must have sensitivity $\Omega(n)$ even on satisfiable instances, and that $(1-\varepsilon)$-approximation algorithms for Max Cut on general graphs incur sensitivity $\Omega(n)$ as well.
These results again improve upon the prior sensitivity bound of $\Omega(n^\delta)$~\cite{fleming2026sensitivity}.

To derive a sensitivity bound specific to Max Cut on bipartite graphs, we construct an LTC whose
associated CSP is precisely Max Cut and substitute it into \Cref{thm:main-intro}. The result is
striking because it holds even on bipartite inputs, where a perfect cut always exists:
\begin{theorem}\label{thm:max-cut-intro}
    There exists a constant $\varepsilon_0>0$ such that for every even $n$ and every $0<\varepsilon\le \varepsilon_0/\sqrt{n}$, any $(1-\varepsilon)$-approximation algorithm for Max Cut on bipartite graphs on $n$ vertices must have sensitivity $\Omega(1/\varepsilon)$.
\end{theorem}
Thus even on this highly structured class, where the optimum cut is trivial, any near-optimal
algorithm must make $\Omega(1/\varepsilon)$ changes after a single edge deletion. As we mentioned,
it was known that any algorithm for $2$-coloring a bipartite graph requires $\Omega(n)$
sensitivity~\cite{varma2023average}, which matches the regime $\varepsilon < 1/n$ in
\Cref{thm:max-cut-intro}. Our result generalizes this by showing an $\Omega(1/\varepsilon)$
sensitivity bound for all $\varepsilon$ as large as $\Theta(1/\sqrt{n})$.

\paragraph{Sensitivity Lower Bounds for Other Combinatorial Problems.}

Applying the standard FGLSS reduction~\cite{feige1996interactive} together with the serial repetition to \Cref{thm:e3lin2-intro}, as in~\cite{fleming2026sensitivity}, yields the following lower bound for the maximum clique problem. 
Here, we consider the setting that the algorithm outputs a vertex set that forms a clique, and sensitivity is measured with respect to deleting a single edge.
\begin{theorem}\label{thm:clique-intro}
    Let $\varepsilon > 0$.
    Then, any algorithm for the maximum clique problem that outputs an $n^{-\varepsilon}$-approximate clique with probability $1-O(1/n)$ must have sensitivity $\Omega(n^{1-c\varepsilon})$ for some universal constant $c>0$.
\end{theorem}
Consider the naive (and inefficient) strategy that computes the maximum clique and then truncates it to size at most $n^{1-\varepsilon}$. 
This procedure achieves an $n^{-\varepsilon}$ approximation with sensitivity $n^{1-\varepsilon}$, and \Cref{thm:clique-intro} shows that this dependence is essentially optimal.
Also, the previous lower bound~\cite{fleming2026sensitivity} shows that there exist \emph{some} constants $\varepsilon,\delta>0$ such that any $n^{-\varepsilon}$-approximation algorithm has sensitivity $\Omega(n^\delta)$.
By contrast, our lower bound of $\Omega(n^{1-c\varepsilon})$ holds for \emph{any} $\varepsilon > 0$.

We also study the \emph{maximum $k$-coverage problem}, where an instance consists of a universe $U$ and a family $\mathcal F$ of subsets of $U$.
The goal is to select a subfamily $\mathcal S \subseteq \mathcal F$ of $k$ sets that covers as many elements as possible.
Sensitivity is measured with respect to deleting a single element from $U$ (and removing it from every set in $\mathcal F$).
Via a reduction from Max E3SAT we derive the following:
\begin{theorem}\label{thm:coverage-intro}
    There exists $\varepsilon_0 > 0$ such that, for every $0 < \varepsilon < \varepsilon_0$ and $k < 1/(2\varepsilon)$, any $(1-\varepsilon)$-approximation algorithm for the maximum $k$-coverage problem must have sensitivity $\Omega(k)$, even when the instance admits a solution that covers all elements with $k$ sets.
\end{theorem}
Observe that any algorithm for the maximum $k$-coverage problem has sensitivity at most $k$, since its output contains at most $k$ sets. 
Thus \Cref{thm:coverage-intro} shows that even permitting a slight approximation does not improve over this trivial upper bound.
We also note that, combining the prior $\Omega(n^\delta)$ lower bound for Max E3SAT~\cite{fleming2026sensitivity} with the same reduction yields only $\Omega(k / n^{1-\delta})$, which is negligible.
Achieving a linear lower bound is therefore essential.

\paragraph{Boolean-Function Consequences.}
The LTC construction behind \Cref{thm:e3lin2-intro} also gives a Boolean-function barrier for satisfiable Max E3SAT. Under the standard reduction, each message $m \in \{0,1\}^k$ indexes a formula $\Phi(m)$ in an explicit satisfiable family, and flipping one message bit changes only $O(1)$ clauses.
Composing a randomized approximation algorithm with the map $m \mapsto \Phi(m)$ and then fixing its internal randomness turns each output bit into a Boolean function on $\{0,1\}^k$. The \emph{influence} of an output bit is the expected number of message coordinates whose flip changes that bit, equivalently, the expected number of adjacent formulas $\Phi(m)$ and $\Phi(m^{\oplus i})$ on which that output bit differs. Our sensitivity lower bound therefore implies that, averaged over the internal randomness, linearly many output coordinates have influence $\Omega(k)$:

\begin{theorem}[Informal version of Corollary~\ref*{cor:sat-deterministic-cumulative-influence}]
There is an explicit bounded-occurrence family of satisfiable Max E3SAT formulas on
$n=\Theta(k)$ variables such that, for every such randomized algorithm, averaged over the
internal randomness, linearly many output bits have influence $\Omega(k)$. 
\end{theorem}

Friedgut's junta theorem and Hatami's pseudo-junta theorem~\cite{Friedgut98,Hatami12} show that small total
influence forces strong structural simplification. Our theorem therefore implies that linearly many
output bits stay a constant distance away from all small juntas and pseudo-juntas
(\Cref{cor:sat-anti-junta}). It also implies that linearly many output
coordinates have linear decision-tree depth (\Cref{cor:sat-dt-depth}). Likewise, average-sensitivity bounds of Boppana and Rossman~\cite{Boppana97,Rossman18} yield lower bounds for
randomized polynomial-size constant-depth circuits (that is, randomized $AC^0$) and randomized
bounded-depth formulas (\Cref{cor:sat-witness-circuit-lb}).

Unlike previous lower bounds for random $k$-SAT, including the bounded-depth \emph{decision} lower bound of Gamarnik, Mossel, and Zadik and the low-degree-polynomial \emph{search} lower bound of Bresler and Huang \cite{GamarnikMosselZadik2023,BreslerHuang2021}, these lower bounds apply to an explicit satisfiable Max E3SAT family and rule out randomized $AC^0$ circuits and bounded-depth formulas for producing near-satisfying assignments even when the algorithm may output any of many near-optimal assignments.

\paragraph{Locality Lower Bounds for the Non-signaling Model.}

Our sensitivity bounds remain strong even in the non-signaling model, which strictly generalizes
classical $\mathsf{LOCAL}$, quantum-$\mathsf{LOCAL}$, and bounded-dependence models: satisfiable
Max E3LIN2 still requires locality $\Omega(\log n)$, and bipartite Max Cut requires locality
$\Omega(\log(1/\varepsilon))$.

In the \emph{non-signaling model}~\cite{akbari2025online,coiteux2024no,arfaoui2014can,gavoille2009can}, given a graph $G=(V,E)$, we can produce an arbitrary output distribution as long as it does not violate the non-signaling principle: for any set of vertices $S \subseteq V$, modifying the structure of the input graph at more than a distance $t$ from $S$ does not affect the output distribution of $S$.
The parameter $t$ is called the \emph{locality} of the model.
This model is stronger than any ``physical'' synchronous distributed computing model, and in particular, it contains classical $\mathsf{LOCAL}$~\cite{linial1992locality}, quantum-$\mathsf{LOCAL}$~\cite{gavoille2009can}, and the bounded-dependence model~\cite{holroyd2017finitary}.
Consequently, any locality lower bound in the non-signaling model applies automatically to these models. 

The following observation from~\cite{fleming2026sensitivity} connects locality to sensitivity.
Suppose we generate a distribution in the non-signaling model for a graph problem with locality $t$.
Sampling from this distribution and outputting the result yields an algorithm whose sensitivity is $O(\Delta^t)$, where $\Delta$ denotes the maximum degree, because deleting an edge can affect the marginal distributions of at most $\Delta^t$ vertices.
Because the CSP instance produced by the $c^3$-LTC~\cite{dinur2022locally,panteleev2022asymptotically} has constant degree, \Cref{thm:e3lin2-intro} translates into the following lower bound:
\begin{theorem}\label{thm:max-e3lin2-non-signaling}
    There exists $\varepsilon > 0$ such that any $(1-\varepsilon)$-approximate non-signaling distribution for Max E3LIN2 must have locality $\Omega(\log n)$, even on satisfiable instances.
\end{theorem}
Marwaha and Hadfield~\cite{marwaha2022bounds} showed that achieving a $(1-\varepsilon)$-approximation on satisfiable instances already requires $\Omega(\log n)$ rounds in quantum-$\mathsf{LOCAL}$, and Fleming and Yoshida~\cite{fleming2026sensitivity} proved the same locality bound in the non-signaling model for general instances. Our theorem unifies these results by establishing that even satisfiable Max E3LIN2 instances demand $\Omega(\log n)$ locality for non-signaling distributions.

Similarly, \Cref{thm:max-cut-intro} yields the following lower bound for the non-signaling model:
\begin{theorem}\label{thm:max-cut-non-signaling}
    There exists a constant $\varepsilon_0>0$ such that for every even $n$ and every $0<\varepsilon\le \varepsilon_0/\sqrt{n}$, any $(1-\varepsilon)$-approximate non-signaling distribution for the maximum cut problem on bipartite graphs must have locality $\Omega(\log (1/\varepsilon))$.
\end{theorem}
Hence setting $\varepsilon = n^{-\delta}$ for any small constant $\delta > 0$ forces locality $\Omega(\log n)$.
The previously best locality barrier on bipartite graphs is that any \emph{one-round} quantum-$\mathsf{LOCAL}$ algorithm achieves at most a $(\frac{1}{2}+\frac{1}{\sqrt{2d}})$-approximation on $d$-regular bipartite graphs~\cite{barak2022classical}.
No stronger locality lower bound was known that rules out $(1-\varepsilon)$-approximations on bipartite inputs, even for $\mathsf{LOCAL}$.  

On general graphs, there exist $(1-\varepsilon)$-approximation algorithms in $\mathsf{LOCAL}$ (and even $\mathsf{CONGEST}$) that run in $O(\mathrm{poly}(1/\varepsilon)\cdot \log n)$ rounds~\cite{censor2017fast}.
By contrast, Fleming and Yoshida~\cite{fleming2026sensitivity} showed an $\Omega(\log n)$ locality lower bound in the non-signaling model for achieving $(1-\varepsilon)$ approximation for some $\varepsilon > 0$.

\subsection{Proof Overview}
We first sketch the argument underlying \Cref{thm:main-intro} and then outline how it yields \Cref{thm:e3lin2-intro,thm:max-cut-intro}.

\paragraph{An analogy between optimization problems and systematic LTCs.}
We exploit a tight parallel between optimization problems and systematic LTCs, where \emph{systematic} means that the encoder copies the $k$ message symbols verbatim into designated coordinates, inducing a canonical input/output split.  
We write $x \sim x'$ for neighboring instances.

Consider an optimization problem with input space $\mathcal I$, output space $\mathcal O$, and evaluation map $e : \mathcal I \times \mathcal O \to [0,1]$.  
A deterministic $(1-\varepsilon)$-approximation algorithm $\mathcal A : \mathcal I \to \mathcal O$ satisfies $e(x,\mathcal A(x)) \in [1-\varepsilon,1]$ for every $x$.  Attaching the input to the output, $\tilde{\mathcal A}(x) = (x,\mathcal A(x))$, yields the diagram
\[
\mathcal I \xrightarrow{\;\tilde{\mathcal A}\;} \mathcal I \times \mathcal O \xrightarrow{\;e\;} [1-\varepsilon,1].
\]
The sensitivity of $\mathcal A$ is the maximum, over neighboring inputs $x\sim x'$, of the distance between $\tilde{\mathcal A}(x)$ and $\tilde{\mathcal A}(x')$ inside $\mathcal I\times\mathcal O$.
To derive a lower bound we therefore search for neighboring inputs that force any $(1-\varepsilon)$-approximation algorithm to output points that are well separated for at least two inputs in $\mathcal I\times\mathcal O$, mirroring the coding-theoretic requirement that distinct messages map to distant codewords.

Now replace the algorithm by a systematic LTC $C$ with an encoder $f$.
Let $n$ be the block length, $k$ be the message length, $\delta$ be the relative distance, and $T$ be the tester with detection probability $\kappa$.
Then, define $e(x) = \Pr{T \text{ accepts } x}$. 
Every codeword $f(m)$ for a message $m$ satisfies $e(m, f(m)) = 1$, so we obtain the analogous diagram
\[
\Sigma^k \xrightarrow{\;f\;} \Sigma^k \times \Sigma^{n-k} \xrightarrow{\;e\;} \set{1}.
\]
Because the code has relative distance $\delta$, distinct messages map to codewords whose relative Hamming distance is at least $\delta$, mirroring the separation requirement for optimization algorithms.

To reason about approximate outputs we invoke strong testability.
Think of $f$ as an exact algorithm that always attains the optimum $1$ of $e$.
If another encoder $\mathcal B$ satisfies $e(m,\mathcal B(m)) \ge 1-\varepsilon$, strong testability guarantees that $\tilde{\mathcal B}(m)$ lies within distance $\varepsilon/\kappa$ of some codeword.
Consequently we get the diagram
\[
\Sigma^k \xrightarrow{\;\tilde{\mathcal B}\;} \Sigma^k \times \Sigma^{n-k} \xrightarrow{\;e\;} [1-\varepsilon,1].
\]
The distance between the codewords of neighboring messages $m\sim m'$ therefore forces a comparable separation between $\tilde{\mathcal B}(m)$ and $\tilde{\mathcal B}(m')$, yielding the desired sensitivity lower bound.

The systematic nature of the code is crucial: it embeds the message directly into the codeword, so the tester $T$ enforces consistency between input and output rather than merely checking membership in the code.

\paragraph{Any approximate encoder has high sensitivity.}
With this analogy in place we outline how to prove the sensitivity lower bound for approximation algorithms on the induced CSP, thereby establishing \Cref{thm:main-intro}.
Fix a randomized $(1-\varepsilon)$-approximation algorithm $\mathcal A$ for $\mathsf{MaxCSP}(C)$.
Given a message $m$, we construct a CSP instance $\mathcal I_C(m)$ that is satisfiable exactly by the codeword for $m$, and strong testability ensures that the output of $\tilde{\mathcal A}$ on $\mathcal I_C(m)$ lies within Hamming distance $\varepsilon/\kappa$ of that codeword.
Hence we may view $\tilde{\mathcal A}$ as an approximate encoder that attempts to extend $m$.
To force large sensitivity we must find neighboring messages $m\sim m'$ whose outputs fall into balls centred at different codewords.
To this end we introduce the nearest-message decoder $\sigma:\Sigma^n \to \Sigma^k$ and show that for some $m\sim m'$ the values $\sigma\circ\tilde{\mathcal A}(\mathcal I_C(m))$ and $\sigma\circ \tilde{\mathcal A}(\mathcal I_C(m'))$ differ with constant probability; by the code’s relative distance this forces $\tilde{\mathcal A}(\mathcal I_C(m))$ and $\tilde{\mathcal A}(\mathcal I_C(m'))$ to have large earth mover's distance.

However, it is not enough to focus on a single pair $m \sim m'$: the encoder $\tilde{\mathcal A}$ might always output the codeword for $m$, achieving zero sensitivity while staying nearly optimal for $m$ and $m'$.
Instead we average over all such pairs.
For any approximate encoder $\tilde{\mathcal A}$ we define a functional $L$ that averages the earth mover's distance between $\tilde{\mathcal A}(\mathcal I_C(m))$ and $\tilde{\mathcal A}(\mathcal I_C(m'))$ across all $m\sim m'$.
This functional lower bounds the worst-case sensitivity and, more strongly, the average sensitivity under uniformly random constraint deletions.

The difficulty is that $\mathcal A$ can otherwise be arbitrary.
Yet $L$ is convex, so we may symmetrize $\tilde{\mathcal A}$ by averaging over permutations of message coordinates and alphabet symbols without increasing $L$.
Composing the symmetrized encoder $\bar{\mathcal A}$ with the nearest-message decoder $\sigma$ produces, for each message $m$, a distribution $\nu_m^*$ over $\Sigma^k$ whose probabilities depend only on the Hamming distance to $m$, making the analysis tractable.
We then show that some neighboring pair $m\sim m'$ forces $\sigma\circ\bar{\mathcal A}(\mathcal I_C(m))$ and $\sigma\circ\bar{\mathcal A}(\mathcal I_C(m'))$ to differ with constant probability, yielding the desired sensitivity lower bound.

\paragraph{Sensitivity lower bounds from explicit LTCs}
Given the connection between sensitivity lower bounds and LTCs, we instantiate our framework with two explicit families.
First, we adapt the left-right Cayley complex (LRCC) codes of Dinur, Evra, Livne, Lubotzky, and Mozes~\cite{dinur2022locally} into systematic LTCs whose induced CSP is Max E3LIN2.
This transformation preserves their rate, distance, and detection probability, letting us transfer the $c^3$-LTC guarantees and derive $\Omega(n)$ average-sensitivity lower bounds for every $(1-\varepsilon)$-approximation algorithm, thereby proving \Cref{thm:e3lin2-intro}.

Second, we analyze signed repetition codes that duplicate each message coordinate alongside its complement.
We partition the coordinates into $A$ (original copies) and $B$ (flipped copies), embed a bipartite expander on $A\cup B$, and use a two-query tester that checks random edges for disagreement, thereby realizing a Max Cut instance.
The replication factor controls both the rate and distance, and feeding this construction into our framework establishes \Cref{thm:max-cut-intro}.

\subsection{Related Work and Discussion}

\paragraph{Max Cut.}
It is natural to ask whether one can beat the sensitivity lower bound of \Cref{thm:max-cut-intro} by applying \Cref{thm:main-intro} to a different two-query LTC. 
However, in such a tester, every constraint enforces that a pair of queried coordinates differs, so codewords correspond exactly to 2-colorings of the tester's constraint graph. The number of codewords, and hence the rate, is dictated by the number of connected components, while the minimum distance is governed by the sizes of those components. 
Consequently, an LTC obtained this way cannot simultaneously achieve high rate and large distance, preventing stronger lower bounds via this route.

\paragraph{Lower Bounds for Other CSPs.}
For a linear code $C$, $\mathsf{MaxCSP}(C)$ necessarily collapses to a system of linear equations.
Attempts to extend our lower bounds via gadget reductions such as pp-definitions or pp-interpretations~\cite{bulatov2008recent,barto2017polymorphisms} stall here: starting from linear equations one cannot reach target languages whose expressive power is incomparable with (or weaker than) linear systems, so reductions to templates like 2SAT or Horn-SAT are ruled out.
Another idea is to replace LTCs with dictatorship tests from the PCP literature~\cite{haastad2001some,khot2007optimal}, but it remains unclear whether such tests can be adapted to yield sensitivity lower bounds.

\paragraph{Machine Unlearning}
\emph{(Exact) machine unlearning} is often formalized by requiring that, after deleting a data item, the algorithm’s output matches the model one would obtain by retraining on the pruned dataset~\cite{guo2020certified,bourtoule2021machine,nguyen2025survey}. 
Viewed through this lens, our sensitivity lower bounds translate directly into recourse lower bounds for unlearning. 
In particular, \Cref{thm:coverage-intro} shows that for the maximum $k$-coverage problem—a canonical abstraction for summarization and selection tasks—the deletion of a single universe element forces any exact-unlearning mechanism that maintains a $(1-\varepsilon)$-approximation to modify $\Omega(k)$ members of the selected set family to match the retrained solution. 
Hence exact unlearning cannot rely on small, local ``patches'': even a single deletion necessitates wholesale restructuring of the learned summary.

\paragraph{Sensitivity on Planar Graphs}
The constraint graph of a 2CSP instance has one vertex per variable and an edge for each constraint connecting the variables it involves.
Our results establish lower bounds for the class of constraint graphs induced by the LTCs we employ; for instance, the repetition code we use for Max Cut yields an expander constraint graph.

A natural question is whether algorithms with low (average) sensitivity exist for non-expanding graphs such as planar ones.
When the planar graph also has bounded degree, the partition oracle~\cite{hassidim2009local,kumar2022random} gives a $(1-\varepsilon)$-approximate Max Cut algorithm with average sensitivity $\mathrm{poly}(1/\varepsilon)$ under vertex deletions.
However, it remains unclear whether an analogous guarantee holds for general planar graphs, where the maximum degree may be unbounded.
For the lower-bound side, an LTC whose Tanner graph is planar and has rate $\geq 5/8$ must have distance $O(1/n)$~\cite{srinivasan2012codes}.
Consequently, applying \Cref{thm:main-intro} to planar graphs would require LTCs of smaller rate.

\paragraph{Overlap Gap Property.}

The \emph{overlap gap property (OGP)} is a geometric/topological description of near-optimal solution spaces: in many random optimization problems, the set of $(1-\varepsilon)$-optimal solutions splits into clusters whose pairwise overlaps avoid a nontrivial interval. 
OGP has emerged as a unifying explanation for barriers for local algorithms in high-dimensional random structures~\cite{gamarnik2014limits,chen2019suboptimality} and is surveyed by Gamarnik, who frames it as a topological obstruction to optimization~\cite{gamarnik2021overlap}.
Classical OGP analyzes overlaps between near-optimal solutions on the same instance (or across randomly coupled instances), whereas our sensitivity framework compares outputs of algorithms on two almost-identical instances.
At a high level the direction is similar: both aim to rule out stable/local algorithms, i.e., OGP via geometric gaps in the near-optimal solution space, and our results via explicit worst-case instances where tiny input perturbations force large output changes.

\paragraph{Circuit-complexity lower bounds for SAT and related search problems.}
For unrestricted circuits, lower bounds for \emph{search}-SAT are closely connected to lower bounds for \emph{decision}-SAT: Pich proved implications between formalized search and decision lower-bound principles \cite{Pich2015}. 
Because such general lower bounds remain out of reach, known assumption-free results mostly concern restricted models, on both the decision and the search side.
On the decision side, Gamarnik, Mossel, and Zadik proved that depth-$d$ bounded-depth circuits for random $2$-SAT near the satisfiability threshold require size $\exp(n^{\Theta(1/d)})$~\cite{GamarnikMosselZadik2023}.
On the search side, Bresler and Huang showed that low-degree polynomial algorithms cannot find satisfying assignments for random $k$-SAT at clause densities of order $2^k \log k / k$, excluding a broad family of algorithms, including local methods and several message-passing procedures~\cite{BreslerHuang2021}.
In a different restricted model, G\"o\"os, Kamath, Robere, and Sokolov proved exponential monotone circuit lower bounds for a monotone variant of XOR-SAT~\cite{GoosKamathRobereSokolov2019}.
By contrast, our results give a lower bound for randomized bounded-depth circuits and formulas for producing near-satisfying assignments on an explicit satisfiable Max E3SAT family, rather than an average-case decision lower bound for random formulas or a lower bound in a different restricted search model.

\section{Preliminaries}\label{sec:prelim}

For a positive integer $n$, let $[n]$ denote the set $\{1,2,\ldots,n\}$. 
For two strings $x,y \in \Sigma^n$, we write $x \sim y$ when they differ on exactly one coordinate.
Also, we write $d_{\mathrm H}(x,y) := \#\{i \in [n] : x_i \neq y_i\}$ for the \emph{(unnormalized) Hamming distance}
and $\delta_{\mathrm H}(x,y):= d_{\mathrm H}(x,y)/n$ for its normalized/relative version.
Note that we use the relative Hamming distance only in the context of error-correcting codes. 
To avoid tedious transforms between normalized Hamming distances with different denominators, 
we always use the block length $n$ of the code as the denominator when dealing with codes.

\subsection{Sensitivity}

Let $(\Omega,d)$ be a finite metric space and let $\mu,\nu$ be distributions over $\Omega$. 
The \emph{earth mover's distance} between $\mu$ and $\nu$ is
\[
    \EMD_d(\mu,\nu)
    := \inf_{\pi\in\Pi(\mu,\nu)} \sum_{x,y\in\Omega} d(x,y)\,\pi(x,y),
\]
where $\Pi(\mu,\nu)$ denotes the set of all couplings of $\mu$ and $\nu$.
When the metric is clear from context we write simply $\EMD(\mu,\nu)$.
If $d$ is the discrete metric on $\Omega$, i.e., $d(x, y) = 1$ when $x\neq y$ and $0$ otherwise, then $\EMD_d(\mu,\nu) = \TV(\mu,\nu)$, the total variation distance between $\mu$ and $\nu$.

Consider a problem in which an input instance $I$ consists of elements $e_1,\ldots,e_n$ and the output space is a metric space $(\Omega,d)$.
For an element $e \in I$, let $I - e$ denote the instance obtained by deleting $e$.
The \emph{sensitivity} of a (possibly randomized) algorithm $\mathcal A$ at $I$ is
\begin{align*}
  \mathsf{Sens}(\mathcal A,I) &:= \max_{i \in [n]} \EMD_d(\mathcal{A}(I), \mathcal{A}(I-e_i)), 
\end{align*}
where we identify $\mathcal A(I)$ with its output distribution.

\subsection{Locally Testable Codes}

For $w\in\Sigma^n$ and a subset $C\subseteq \Sigma^n$, we set $\delta_{\mathrm H}(w, C) := \min_{c\in C} \delta_{\mathrm H}(w, c)$.

\begin{definition}[Error-correcting code]\label{def:ecc}
    Let $n,k\in\mathbb N$ with $n\ge k$ and let $\Sigma$ be a finite alphabet.
    A \emph{(block) error-correcting code} of length $n$ and message length $k$ over $\Sigma$ is a pair $(C,f)$, where $C\subseteq \Sigma^n$ and the encoder $f:\Sigma^k\to C$ is a bijection.
    The \emph{rate} of the code is $r=k/n$, and its \emph{(relative) distance} is $\min_{c,c'\in C, c\neq c'} \delta_{\mathrm H}(c,c')$.
    We often identify the code with the subset $C$ when the encoder is clear from context.
\end{definition}

\begin{definition}[Linear code]
    A code $(C,f)$ over $\mathbb F_p$ is \emph{linear} if $C$ is a linear subspace of $\mathbb F_p^n$, equivalently if $f$ is a linear map.
    The (absolute) distance of a linear code equals the minimum Hamming weight of a nonzero codeword.
\end{definition}

\begin{definition}[Locally testable codes]\label{def:ELTC}
    Let $r, \delta, \kappa\in \mathbb R_{\ge 0}$, $q, \theta\in \mathbb N$.
    We say an error-correcting code $(C,f)$ is a \emph{locally testable code} (LTC) with detection probability $\kappa$, $q$ queries, and variable degree $\theta$ if the following properties hold:
    \begin{itemize}
        \itemsep=0pt
        \item There exists a collection $\mathcal S$ of $q$-subsets of $[n]$ such that each subset $S$ is 
              associated with a set $V_S\subseteq \Sigma^S$ called \emph{allowed local views},
              satisfying $c|_{S}\in V_S$ for all $c \in C$, and
                \[
                \forall w\in \Sigma^n:\; \Pr[S\sim \mathcal S]{w|_S\notin V_S} \ge \kappa\cdot \delta_{\mathrm H}(w, C),
                \]
                where $S \sim \mathcal S$ denotes the uniform distribution over $\mathcal S$.
        \item For any $i\in [n]$, we have $\sum_{S\in\mathcal S} \1{i\in S}\le\theta$.
    \end{itemize}
    When we need to emphasize the local tests, we denote the LTC by $(C, f, \mathcal S, \mathcal V)$, where $\mathcal V = \{V_S\}_{S \in \mathcal S}$.
\end{definition}

\begin{definition}[Systematic code]\label{def:EECC}
    Let $n,k\in\mathbb{N}$ with $n\ge k$ and let $\Sigma$ be a finite alphabet. 
    Let 
    $\pi:\Sigma^n\to\Sigma^k$ be the projection onto the first $k$ coordinates.
    A \emph{systematic code} of length $n$ and message length $k$ over $\Sigma$ is an error-correcting code $(C,f)$
    with an encoder $f:\Sigma^k\to\Sigma^n$ such that $\pi\circ f=\mathrm{id}_{\Sigma^k}$.
    Thus the message $x$ appears verbatim as the first $k$ entries of the codeword $f(x)$.
    We also say that a code $C\subseteq\Sigma^n$ is \emph{systematic} if it admits such an encoder.
\end{definition}
The next lemma shows that any linear locally testable code becomes systematic after permuting coordinates.
The proof is deferred to \Cref{sec:lin-sys}.
\begin{lemma}\label{lem:lin-sys}
    If $C \subseteq \mathbb F_p^n$ is a linear LTC, then there exists a permutation matrix $\Pi$ such that $\Pi C$ is a systematic LTC with the same parameters.
\end{lemma}

\subsection{Constraint Satisfaction Problems}

Let $V$ be a finite set of variables and let $\Sigma$ be a finite alphabet.
An instance of a \emph{constraint satisfaction problem (CSP)} over $\Sigma$ is a tuple $(V, \Sigma, \mathcal C)$, where $\mathcal C$ is a finite collection of constraints.
Each constraint is given by a pair $(S, R)$ with $S \subseteq V$ its scope and $R \subseteq \Sigma^{S}$ the set of allowed local assignments.
An assignment $\alpha \in \Sigma^{V}$ \emph{satisfies} $(S, R)$ if $\alpha|_{S} \in R$, and it is a solution to the instance when it satisfies every constraint in $\mathcal C$.
The associated \emph{maximum constraint satisfaction problem (Max CSP)} asks for an assignment $\alpha \in \Sigma^{V}$ that maximizes the number of satisfied constraints.

For a CSP instance $\mathcal I = (V, \Sigma, \mathcal C)$ and a constraint $c\in \mathcal C$, we write $\mathcal I - c$ for the instance obtained by removing $c$ from $\mathcal C$.

\section{Sensitivity Lower Bounds via LTCs}\label{sec:lower_bound}
The goal of this section is to prove \Cref{thm:main-intro}.
We organize the argument in four stages.
First, \Cref{sec:csp-ltc} shows how a systematic LTC induces a CSP whose constraints mirror the tester.
Next, in \Cref{sec:struct-tools} we discuss the nearest-message decoder, which recovers the planted message from near-optimal assignments, and accompanying stability bounds.
Building on these tools, \Cref{sec:tv-distance} symmetrizes the decoder's output distributions and proves a large total-variation gap between neighboring messages.
Finally, \Cref{sec:worst-case} converts this gap into an earth mover's distance lower bound for worst-case sensitivity, and the same argument also yields an average-sensitivity analogue, giving \Cref{thm:main-intro}.
\subsection{CSP Induced by an LTC}\label{sec:csp-ltc}
For $m \in \Sigma^k$ and a subset $U \subseteq \Sigma^{n-k}$, we write
$m \oplus U := \{(m, u) : u \in U\} \subseteq \Sigma^{n}$.
\begin{definition}\label{def:CSP-LTC}
    For a systematic LTC $(C, f, \mathcal S, \mathcal V = \{V_S\}_{S \in \mathcal S})$ and a message $m\in \Sigma^k$, we construct a CSP instance $\mathcal I_{C,f,\mathcal S, \mathcal V}(m) = (V,\Sigma,\mathcal C)$ as follows:
    \begin{itemize}
    \itemsep=0pt
    \item $V := \{x_i\}_{i=k+1}^n$.
    \item 
    For every $S\in\mathcal S$, let $S' := S \cap \{k+1,\dots,n\}$ and define
    \[
    R_S := \left\{\alpha \in \Sigma^{S'} \mid \exists \beta \in V_S: \beta|_{[k]} = m \wedge \beta|_{S'} = \alpha\right\}.
    \]
    The constraint corresponding to $S$ is $(S', R_S)$.
    \item $\mathcal C := \{(S', R_S) : S\in\mathcal S\}$.
    \end{itemize}
    For brevity, we write $\mathcal I_{C}(m)$ for $\mathcal I_{C,f,\mathcal S, \mathcal V}(m)$.
\end{definition}

Intuitively, the message coordinates of a systematic codeword are fixed once we choose $m$, so the only freedom left is the parity part. 
The tester’s local view on $S$ allows exactly the projections of valid codewords onto $S$. 
Intersecting $V_S$ with $m \oplus \Sigma^{n-k}$ therefore isolates those local patterns that agree with $m$ on the message positions, and projecting them onto the parity indices $S'$ yields the admissible assignments for the remaining variables. 
The family $\mathcal C$ is just the collection of these derived local constraints, so any assignment to the parity coordinates satisfies the CSP instance precisely when it extends $m$ to a full codeword in $C$.

\begin{definition}
    For a systematic LTC $C$, we define an optimization problem $\MaxCSP(C)$ as follows.
    The input is a message $m\in \Sigma^k$, and the goal is to compute an assignment in $\Sigma^V$ that maximizes the number of satisfied constraints in $\mathcal I_C(m)$.
\end{definition}
Notice that all constraints in $\mathcal I_C(m)$ can be satisfied by assigning the code bits of $f(m)$
to the variables.

Throughout this section we reuse the convention from the introduction: for any approximate algorithm $\mathcal A$ for $\MaxCSP(C)$ we define the randomized map
\[
    \tilde{\mathcal A}(\mathcal I_C(m)) := (m, \mathcal A(\mathcal I_C(m)))
\]
which simply pairs the input message with the assignment returned by $\mathcal A$.

\subsection{Nearest-Message Decoder}\label{sec:struct-tools}
We collect here the decoder and stability statements that connect near-optimal assignments to nearby codewords.

\begin{definition}
    Let $\sigma: \Sigma^n \to \Sigma^k$ be the \emph{nearest-message decoder} for $(C,f)$, that is,
    for every $w\in\Sigma^n$ the value $\sigma(w)$ belongs to
    \[
        \argmin_{m'\in\Sigma^k} \delta_{\mathrm H}\left(f(m'), w\right).
    \]
    We break ties arbitrarily.
\end{definition}

\begin{lemma}\label{lem:m-approximation}
    Let $C$ be a systematic LTC with rate $r = k/n$ and let $\mathcal A$ be an $(1-\varepsilon)$-approximation randomized algorithm for $\MaxCSP(C)$.
    Then for every $m\in\Sigma^k$,
    \[
        \mathbb E\bigl[\delta_{\mathrm H}(\sigma(\tilde{\mathcal A}(\mathcal I_C(m))), m)\bigr] \le \frac{\varepsilon}{\kappa}.
    \]
\end{lemma}
\begin{proof}
    Because $\mathcal A$ achieves approximation ratio $1-\varepsilon$, the expected fraction of satisfied constraints on the instance $\mathcal I_C(m)$ is at least $1-\varepsilon$, whence
    \[
        \mathbb E\left[\frac{\#\text{ of unsatisfied constraints in }\mathcal I_C(m)}{\lvert\mathcal C\rvert}\right] \le \varepsilon.
    \]
    For every word $w\in\Sigma^n$, \Cref{def:ELTC} guarantees that the rejection probability is at least $\kappa\cdot\delta_{\mathrm H}(w,C)$.
    Applying this to the random output $\tilde{\mathcal A}(\mathcal I_C(m))$ gives
    \[
        \frac{\#\text{ of unsatisfied constraints in }\mathcal I_C(m)}{\lvert\mathcal C\rvert} \ge \kappa\cdot\delta_{\mathrm H}(\tilde{\mathcal A}(\mathcal I_C(m)), C).
    \]
    Taking expectations yields $\mathbb E[\delta_{\mathrm H}(\tilde{\mathcal A}(\mathcal I_C(m)), C)] \le \varepsilon/\kappa$.

    Let $w := \tilde{\mathcal A}(\mathcal I_C(m))$ and let $c^\star := f(\sigma(w))$ be the codeword selected by the nearest-message decoder.
    By definition, $\delta_{\mathrm H}(w, C) = \delta_{\mathrm H}(w, c^\star)$.
    Because $f$ is systematic, the first $k$ entries of $w$ coincide with the message $m$, while the first $k$ entries of $c^\star$ coincide with $\sigma(w)$.
    Consequently every coordinate where $\sigma(w)$ and $m$ differ must also be a coordinate where $w$ and $c^\star$ differ.
    We therefore have
    \[
        \delta_{\mathrm H}(\sigma(w), m) \le \delta_{\mathrm H}(w, c^\star) =\delta_{\mathrm H}(w, C).
    \]
    Taking expectations over the randomness of $\tilde{\mathcal A}$ completes the proof.
\end{proof}

\subsection{Total Variation Gap for Adjacent Messages}\label{sec:tv-distance}

In this section, we consider symmetrizing the decoder's output distributions for near-optimal assignments
and show that the output distributions are well separated in total variation in expectation over neighboring messages.

Let
\begin{align*}
    &\mathcal D_\varepsilon :=\\
    &\Bigl\{
        (\nu_m)_{m\in\Sigma^k}
        \;\Big|\;
        \exists\,(1-\varepsilon)\text{-approx algorithm }\mathcal A
        \text{ for }\MaxCSP(C) 
        \text{ with }\nu_m = \mathrm{law}(\sigma\circ \tilde{\mathcal A}(\mathcal I_C(m)))\text{ for all }m
    \Bigr\}.
\end{align*}
Here $\mathrm{law}(X)$ denotes the distribution of a random variable $X$ and $\sigma$ is the nearest-message decoder.

\begin{lemma}\label{lem:Deps-convex}
    The set $\mathcal D_\varepsilon$ is convex.
\end{lemma}
\begin{proof}
    Let $(\mu_m)$ and $(\nu_m)$ lie in $\mathcal D_\varepsilon$, witnessed respectively by randomized $(1-\varepsilon)$-approximation algorithms $\mathcal A_1$ and $\mathcal A_2$.
    For $\lambda\in[0,1]$, define an algorithm $\tilde{\mathcal A}_\lambda$ which, on input $m$, runs $\tilde{\mathcal A}_1(\mathcal I_C(m))$ with probability $\lambda$ and $\tilde{\mathcal A}_2(\mathcal I_C(m))$ otherwise, and returns the resulting assignment.
    The expected number of satisfied constraints produced by $\tilde{\mathcal A}_\lambda$ is the convex combination of those for $\tilde{\mathcal A}_1$ and $\tilde{\mathcal A}_2$, so $\tilde{\mathcal A}_\lambda$ is again a $(1-\varepsilon)$-approximation algorithm.
    Moreover,
    $\mathrm{law}(\sigma\circ \tilde{\mathcal A}_\lambda(\mathcal I_C(m))) = \lambda\,\mu_m + (1-\lambda)\,\nu_m$
    for every $m\in\Sigma^k$.
    Hence the family $\left(\lambda\,\mu_m + (1-\lambda)\,\nu_m\right)_{m\in\Sigma^k}$ belongs to $\mathcal D_\varepsilon$, establishing convexity.
\end{proof}

\begin{theorem}\label{thm:closest-message-tv}
    Let $C$ be a systematic LTC and $\mathcal A$ be a randomized $(1-\varepsilon)$-approximation algorithm for $\MaxCSP(C)$.
    Then we have
    \[
    \E[m\in \Sigma^k]{ \E[m'\sim m]{
    \TV\left( \sigma\circ \tilde{\mathcal A}(\mathcal I_C(m)), \sigma\circ \tilde{\mathcal A}(\mathcal I_C(m')) \right)
    }
    }
    \ge 1- \frac{|\Sigma|}{|\Sigma|-1} \cdot \frac{\varepsilon}{r \kappa}
    .
    \]
\end{theorem}
\begin{proof}
    We define $L$ as a functional on $\mathcal D_\varepsilon$ by setting
    \begin{align*}
        L\left( \set{\nu_m}_{m\in \Sigma^k} \right)
        &:=
        \frac{1}{k(|\Sigma|-1)|\Sigma|^{k}}
        \sum_{\substack{m_1, m_2\in \Sigma^k \\ m_1 \sim m_2}}
        \TV\left( \nu_{m_1}, \nu_{m_2} \right)
        \\&=
    \E[m\in \Sigma^k]{ \E[m'\sim m]{
    \TV\left( \sigma\circ \tilde{\mathcal A}(\mathcal I_C(m)), \sigma\circ \tilde{\mathcal A}(\mathcal I_C(m')) \right)
    } }
    .
    \end{align*}
    that is, the average total variation distance between the distributions associated with neighboring messages.
    We want to find an element in $\mathcal D_\varepsilon$ that minimizes $L$.

    Suppose $\set{\nu_m}\in \mathcal D_\varepsilon$ is a minimizer of $L$.
    In the following, we construct another minimizer $\set{\nu^*_m}$ that is symmetric in the sense that $\nu^*_m(m')$ only depends on $\delta_{\mathrm H}(m, m')$.
    To do so, we define two kinds of permutations that act on $\mathcal D_\varepsilon$.
    For a finite set $S$ let $\mathrm{perm}(S)$ denote the collection of all permutations of $S$.
    We then have:
    \begin{enumerate}
        \item For any $\pi\in\mathrm{perm}([k])$, define $\pi(m) = (m_{\pi(1)}, \ldots, m_{\pi(k)})$. Then define $\set{\nu_m^\pi}$ by \[\nu_m^\pi(m') := \nu_{\pi(m)}(\pi(m')).\]
        \item For any $\tau = (\tau_1,\ldots,\tau_k)\in\mathrm{perm}(\Sigma)^k$, define $\tau(m) = (\tau_1(m_1), \ldots, \tau_k(m_k))$. Then define $\set{\nu_m^\tau}$ by \[\nu_m^\tau(m') := \nu_{\tau(m)}(\tau(m')).\]
    \end{enumerate}
    We can also combine the two kinds of permutations. For any $\pi\in\mathrm{perm}([k])$ and $\tau\in\mathrm{perm}(\Sigma)^k$, we define $\set{\nu_m^{\pi,\tau}}$ by \[\nu_m^{\pi,\tau}(m') := \nu_{\tau(\pi(m))}(\tau(\pi(m'))).\]
    Since both $\pi$ and $\tau$ are bijections that preserve the Hamming distance, we have
    \[
        L\left( \set{\nu_m^{\pi,\tau}} \right) =
        L\left( \set{\nu_m} \right).
    \]

    We define $\set{\nu^*_m}$ as the mixture of all $\set{\nu_m^{\pi,\tau}}$:
    \[ \nu^*_m(m') := \frac{1}{(|\Sigma|!)^kk!}\sum_{\tau\in\mathrm{perm}(\Sigma)^k}\sum_{\pi\in\mathrm{perm}([k])} \nu^{\pi,\tau}_{m}(m'). \]
    The mixture $\set{\nu^*_m}$ is in $\mathcal D_\varepsilon$ by \Cref{lem:Deps-convex}. Note that $L$ is convex since $\TV(\cdot, \cdot)$ is convex.
    Thus we have
    \[
     L\left(\set{\nu^*_m}\right) \le
    \frac{1}{(|\Sigma|!)^kk!}\sum_{\tau\in\mathrm{perm}(\Sigma)^k}\sum_{\pi\in\mathrm{perm}([k])}L\left(\set{ \nu^{\pi,\tau}_{m}}\right)
    =
    L\left(\set{\nu_m}\right).
    \]
    So $\set{\nu^*_m}$ is also a minimizer of $L$.
    We may fix two different characters $a_0, a_1\in\Sigma$, and define $\bar m = (a_0, \ldots, a_0)$.
    For any $m, m'\in \Sigma^k$, we can choose $\pi^*\in\mathrm{perm}([k])$ such that $\pi^*(m)$ and $\pi^*(m')$ are different at exactly the first $d_{\mathrm H}(m, m')$ coordinates, and $\tau^*\in\mathrm{perm}(\Sigma)^k$ such that
    \begin{align*}
        \forall 1\le i\le k: \tau^*_i(\pi^*(m)_i) = a_0,\quad
        \tau^*_i(\pi^*(m')_i) =
        \begin{cases}
            a_1, & 1\le i\le d_{\mathrm H}(m, m')\\
            a_0, & d_{\mathrm H}(m, m')< i\le k
        \end{cases}.
    \end{align*}
    Then we have
    \begin{align*}
       \nu^*_m(m') &= 
       \frac{1}{(|\Sigma|!)^kk!}\sum_{\tau\in\mathrm{perm}(\Sigma)^k}\sum_{\pi\in\mathrm{perm}([k])} \nu^{\pi,\tau}_{m}(m')
       \\&=
       \frac{1}{(|\Sigma|!)^kk!}\sum_{\tau\in\mathrm{perm}(\Sigma)^k}\sum_{\pi\in\mathrm{perm}([k])} \nu^{\pi,\tau}_{\tau^*(\pi^*(m))}(\tau^*(\pi^*(m')))
       \\&=
       \nu^*_{\tau^*(\pi^*(m))}(\tau^*(\pi^*(m')))
       \\&=
         \nu^*_{\bar m}\left((\underbrace{a_1, \ldots, a_1}_{d_{\mathrm H}(m, m')}, \underbrace{a_0, \ldots, a_0}_{k - d_{\mathrm H}(m, m')})\right).
    \end{align*}
    Thus $\nu^*_m$ is symmetric in the sense that for any $m, m'\in \Sigma^k$,
    $\nu^*_m(m')$ only depends on $\delta_{\mathrm H}(m, m')$.
    That is, there is a function $g:\{0, 1, \ldots, k\}\to [0, 1]$ such that \[\nu^*_m(m') = g\left(d_{\mathrm H}(m, m')\right).\]

    Now we can lower bound $L\left(\set{\nu^*_m}\right)$.
    Since $\set{\nu^*_m}$ is the distribution constructed from the output of a $(1-\varepsilon)$-approximation algorithm $\mathcal A$ by definition,
    we have \[\E{\delta_{\mathrm H}(\sigma(\tilde{\mathcal A}(\mathcal I_C(m))), m)}\le \frac{\varepsilon}{\kappa}\]
    by \Cref{lem:m-approximation}.

    Pick an arbitrary $m\in \Sigma^k$.
    Notice that $\Pr[m'\sim\nu^*_m]{m'_i = m_i}$ is independent of $i$ because $\nu^*_m$ depends on $m'$ only through its Hamming distance from $m$, and every coordinate plays the same role.
    So we have
    \begin{align*}
    \Pr[m'\sim\nu^*_m]{m'_1 \neq m_1} &= \frac1k\sum_{i\in [k]}
    \Pr[m'\sim\nu^*_m]{m'_i \neq m_i} \\&=
    \frac1k\cdot
    \E{d_{\mathrm H}(\sigma(\tilde{\mathcal A}(\mathcal I_C(m))), m)} 
    \\&=
    \frac{n}{k}\cdot
    \E{\delta_{\mathrm H}(\sigma(\tilde{\mathcal A}(\mathcal I_C(m))), m)} 
    \\&\le \frac{\varepsilon}{r\kappa}.
    \end{align*}
    This means
    \begin{equation}\label{eq:pr-lower-bound}
    \Pr[m'\sim\nu^*_m]{m'_1 = m_1} \ge 1 - \frac{\varepsilon}{r\kappa}.
    \end{equation}

    Similarly, picking a character $a\in\Sigma \setminus \{m_1\}$, we define $\hat m$ as the message $(a, m_2,\cdots, m_k)$, we have
    $\Pr[m'\sim\nu^*_{\hat m}]{m'_1 = a} \ge 1-\frac{\varepsilon}{r\kappa}$,
    which means
    \begin{equation}\label{eq:pr-upper-bound}
        (|\Sigma| - 1)\Pr[m'\sim\nu^*_{\hat m}]{m'_1 = m_1} = 
        \sum_{b\in\Sigma\setminus\{a\}} \Pr[m'\sim\nu^*_{\hat m}]{m'_1 = b}
        =
        1 - \Pr[m'\sim\nu^*_{\hat m}]{m'_1 = a}
        \le
        \frac{\varepsilon}{r\kappa},
    \end{equation}
    and thus
    \[\Pr[m'\sim\nu^*_{\hat m}]{m'_1 = m_1} \le \frac{1}{|\Sigma|-1}\frac{\varepsilon}{r\kappa}.\]

    Putting \Cref{eq:pr-lower-bound,eq:pr-upper-bound} together, we have
    \begin{align*}
        L\left(\set{\nu^*_m}\right) = \TV(\nu^*_m, \nu^*_{\hat m}) \ge
        \left|\Pr[m'\sim\nu^*_m]{m'_1 = m_1} - \Pr[m'\sim\nu^*_{\hat m}]{m'_1 = m_1}\right|
        \ge 1- \frac{|\Sigma|}{|\Sigma|-1} \cdot \frac{\varepsilon}{r \kappa},
    \end{align*}
    and the claim follows.
\end{proof}

\begin{lemma}\label{lem:emd-expectation}
    Let $C$ be a systematic LTC with detection probability $\kappa$ and $\mathcal A$ be a randomized $(1-\varepsilon)$-approximation algorithm for $\MaxCSP(C)$.
    Then we have
    \[
    \E[m\in \Sigma^k]{ \E[m'\sim m]{
    \EMD_{\delta_\mathrm H}\left( \mathcal A(\mathcal I_C(m)), \mathcal A(\mathcal I_C(m')) \right)
    }
    }
    \ge
    \delta\left( 1-\frac{|\Sigma|}{|\Sigma|-1}\frac{\varepsilon}{r \kappa} \right) - \frac{2\varepsilon}{\kappa}
    - \frac1n
    .
    \]
\end{lemma}

\begin{proof}
    By \Cref{thm:closest-message-tv}, we have
    \[
    \E[m\in \Sigma^k]{ \E[m'\sim m]{
        \TV\left( \sigma\circ\tilde{\mathcal A}(\mathcal I_C(m)), \sigma\circ \tilde{\mathcal A}(\mathcal I_C(m')) \right)
    }}
        \ge 1 - \frac{|\Sigma|}{|\Sigma| - 1} \cdot \frac{\varepsilon}{r \kappa}
    .
    \]
    Let $\nu_m$ and $\nu_{m'}$ denote the laws of $\sigma\circ\tilde{\mathcal A}(\mathcal I_C(m))$ and $\sigma\circ\tilde{\mathcal A}(\mathcal I_C(m'))$, respectively.
    Since total variation on $\Sigma^k$ coincides with $\EMD_{d_c}$, where $d_c$ is the discrete measure, every coupling $\pi_k\in\Pi(\nu_m,\nu_{m'})$ satisfies
    \[
        \E[(a,b)\sim\pi_k]{d_c(a,b)} \geq 1 - \frac{|\Sigma|}{|\Sigma|-1}\frac{\varepsilon}{r\kappa}.
    \]
    Let $\mu_m$ and $\mu_{m'}$ be the laws of $\tilde{\mathcal A}(\mathcal I_C(m))$ and $\tilde{\mathcal A}(\mathcal I_C(m'))$.
    By conditioning on the decoded message, we extend any such $\pi_k$ to a coupling $\Pi\in\Pi(\mu_m,\mu_{m'})$ on $\Sigma^n\times\Sigma^n$: draw $(a,b)\sim\pi_k$ and then sample $x\sim\mu_m$ and $y\sim\mu_{m'}$ conditioned on $\sigma(x)=a$ and $\sigma(y)=b$ (the conditionals exist because $\pi_k$ has marginals $\nu_m$ and $\nu_{m'}$).
    Under $\Pi$ we have $(\sigma(x),\sigma(y))=(a,b)$ almost surely.
    For any outputs $x,y\in\Sigma^n$ of $\tilde{\mathcal A}$, the triangle inequality gives
    \[
        \delta_{\mathrm H}(x,y) \ge \delta_{\mathrm H}(f\circ\sigma(x),f\circ\sigma(y)) - \delta_{\mathrm H}(x,f\circ\sigma(x)) - \delta_{\mathrm H}(y,f\circ\sigma(y)).
    \]
    Whenever $d_c(\sigma(x),\sigma(y))=1$, the first term is at least $\delta$, so averaging over $\Pi$ yields
    \[
        \E[(x,y)\sim\Pi]{\delta_{\mathrm H}(x,y)} \ge \delta\,\E[(a,b)\sim\pi_k]{d_c(a,b)} - \E[x\sim\mu_m]{\delta_{\mathrm H}(x,f\circ\sigma(x))} - \E[y\sim\mu_{m'}]{\delta_{\mathrm H}(y,f\circ\sigma(y))}.
    \]
    Lemma~\ref{lem:m-approximation} bounds each error term by $\varepsilon/\kappa$, giving
    \[
        \EMD_{\delta_{\mathrm H}}\bigl(\tilde{\mathcal A}(\mathcal I_C(m)),\tilde{\mathcal A}(\mathcal I_C(m'))\bigr)
        \ge \delta\cdot \TV\left( \tilde{\mathcal A}(\mathcal I_C(m)),\tilde{\mathcal A}(\mathcal I_C(m')) \right) - \frac{2\varepsilon}{\kappa}.
    \]
    By the definition of $\tilde{\mathcal A}$, the first $k$ coordinates of the outputs are fixed to $m$ and $m'$ respectively,
    so we have
    \begin{align*}
        \EMD_{\delta_{\mathrm H}}\bigl({\mathcal A}(\mathcal I_C(m)),{\mathcal A}(\mathcal I_C(m'))\bigr)
        &=
        \EMD_{\delta_{\mathrm H}}\bigl(\tilde{\mathcal A}(\mathcal I_C(m)),\tilde{\mathcal A}(\mathcal I_C(m'))\bigr)
        - \frac1n
        \\&\ge
        \delta\cdot \TV\left( \tilde{\mathcal A}(\mathcal I_C(m)),\tilde{\mathcal A}(\mathcal I_C(m')) \right) - \frac{2\varepsilon}{\kappa} - \frac1n
        .
    \end{align*}

    Averaging over $m$ and $m'$ completes the proof.
\end{proof}

\subsection{Worst-Case Sensitivity Bound}\label{sec:worst-case}

We now convert the total-variation gap into an earth-mover lower bound for neighboring messages.

\begin{theorem}[Worst-case sensitivity via LTC]\label{thm:sensitivity-opt-message}
    Let $(C,f, \mathcal S, \mathcal V)$ be a systematic LTC over $\Sigma$ with rate $r$, distance $\delta$, variable degree $\theta$, and detection probability $\kappa$.
    Suppose $\mathcal A$ is a $(1-\varepsilon)$-approximation algorithm for $\MaxCSP(C)$.
    Then there exist messages $m, m' \in \Sigma^k$ with $m \sim m'$ such that
    \[
        \EMD_{\delta_{\mathrm H}}\!\left(\mathcal A(\mathcal I_C(m)), \mathcal A(\mathcal I_C(m'))\right)
        \ge
        \delta\left(1 - \frac{|\Sigma|}{|\Sigma| - 1}\cdot\frac{\varepsilon}{r \kappa} \right) - \frac{2\varepsilon}{\kappa} - \frac1n
        .
    \]
\end{theorem}

\begin{proof}
The claim follows directly from \Cref{lem:emd-expectation} 
since the expectation over $m$ and $m'\sim m$ implies the existence of some pair $m, m'$ achieving at least the expected value.
\end{proof}

\begin{proof}[Proof of \Cref{thm:main-intro}]
    Let $C \subseteq \mathbb F_p^n$ be a linear LTC with rate $r$, distance $\delta$, and detection probability $\kappa$, and let $\mathcal A$ be a $(1-\varepsilon)$-approximation algorithm for $\mathsf{MaxCSP}(C)$.
    By \Cref{lem:lin-sys}, after a fixed permutation of coordinates we may regard $C$ as a systematic LTC; we henceforth work in this permuted coordinate system and write $f$ for the associated systematic encoder.
    Consequently every instance $\mathcal I_C(m)$ is satisfiable, witnessed by the codeword $f(m)$.
    Applying \Cref{thm:sensitivity-opt-message} yields adjacent messages $m \sim m'$ satisfying
    \[
        \EMD_{\delta_{\mathrm H}}\!\left(\mathcal A(\mathcal I_C(m)), \mathcal A(\mathcal I_C(m'))\right)
        \ge
        \delta\left(1-\frac{p}{p-1}\frac{\varepsilon}{r\kappa}\right) - \frac{2\varepsilon}{\kappa} - \frac1n
        .
    \]
    The two instances differ only on the tester constraints that query the unique message coordinate where $m$ and $m'$ disagree; by \Cref{def:ELTC} their number is at most $\theta$.
    Enumerate these constraints in $\mathcal I_C(m)$ as $c_1,\ldots,c_t$ with $t\le \theta$, and let $c'_j$ denote the corresponding constraint in $\mathcal I_C(m')$ for each $j$.
    For $j=0,\ldots,t$ define the chains of satisfiable instances
    \[
        \mathcal I^{(0)} := \mathcal I_C(m),\qquad
        \mathcal I^{(j)} := \mathcal I^{(j-1)} - c_j,
    \]
    and
    \[
        \mathcal I'^{(0)} := \mathcal I_C(m'),\qquad
        \mathcal I'^{(j)} := \mathcal I'^{(j-1)} - c'_j.
    \]
    Removing these constraints leaves the same residual instance, so $\mathcal I^{(t)} = \mathcal I'^{(t)}$.
    Each $\mathcal I^{(j)}$ is satisfied by $f(m)$ and every $\mathcal I'^{(j)}$ is satisfied by $f(m')$.
    Applying the triangle inequality along the two chains,
    \begin{align*}
        \EMD_{\delta_{\mathrm H}}\!\left(\mathcal A(\mathcal I_C(m)), \mathcal A(\mathcal I_C(m'))\right)
        &\le
        \sum_{j=1}^t
        \EMD_{\delta_{\mathrm H}}\!\left(\mathcal A(\mathcal I^{(j-1)}), \mathcal A(\mathcal I^{(j)})\right)
        +
        \sum_{j=1}^t
        \EMD_{\delta_{\mathrm H}}\!\left(\mathcal A(\mathcal I'^{(j-1)}), \mathcal A(\mathcal I'^{(j)})\right).
    \end{align*}
    Combining this inequality with the lower bound above and recalling $t\le \theta$, there exists an index $j^\star$ such that
    \[
        \max\left\{
        \EMD_{\delta_{\mathrm H}}\!\left(\mathcal A(\mathcal I^{(j^\star-1)}), \mathcal A(\mathcal I^{(j^\star)})\right),
        \EMD_{\delta_{\mathrm H}}\!\left(\mathcal A(\mathcal I'^{(j^\star-1)}), \mathcal A(\mathcal I'^{(j^\star)})\right)
        \right\}
        \ge
        \frac{1}{2\theta}\left(
        \delta\left(1-\frac{p}{p-1}\frac{\varepsilon}{r\kappa}\right) - \frac{2\varepsilon}{\kappa} - \frac1n
        \right).
    \]
    In either case we obtain a satisfiable instance together with the deletion of a single constraint that witnesses the required sensitivity gap.
    Because $\delta_{\mathrm H} = d_{\mathrm H}/n$, multiplying by $n$ converts this to the unnormalized Hamming metric.
    Hence there exists a satisfiable instance witnessing worst-case sensitivity at least
    \[
        \frac{n}{2\theta}\left(
        \delta\left(1-\frac{p}{p-1}\frac{\varepsilon}{r\kappa}\right) - \frac{2\varepsilon}{\kappa}\right) - \frac{1}{2\theta},
    \]
    completing the proof.
\end{proof}

\section{Sensitivity Lower Bounds from Explicit LTCs}\label{sec:construction}
This section instantiates the abstract sensitivity framework developed in \Cref{sec:lower_bound}.
We plug explicit families of locally testable codes into \Cref{thm:main-intro} to derive sensitivity lower bounds for Max E3LIN2 and Max Cut.

\subsection{LRCC Codes and Max E3LIN2}\label{subsec:lrcc-max-e3lin2}

In this section, we present our construction of a systematic LTC whose local tests are parity constraints on constantly many variables.
Our construction is based on the \emph{left-right Cayley complex (LRCC) code} introduced by \cite{dinur2022locally},
with the concrete construction as in Sections 5 and 6 in \cite{dinur2022locally}.

\begin{lemma}[Theorem 1.1 in \cite{dinur2022locally}]\label{lem:thm-c3-code}
    For every $0<r<1$, there exist $\delta, \kappa > 0$, $q\in \mathbb Z_{\>0}$, 
    and an infinite family of LTCs $\set{C_n}$ with rate $r$,
    distance $\delta$, detection probability $\kappa$, and query complexity $q$.
\end{lemma}

\begin{remark}
    Section~5 of \cite{dinur2022locally} shows that the codes in \Cref{lem:thm-c3-code} are linear over the binary field $\mathbb F_2$.
\end{remark}

For completeness we recall the notation surrounding the left-right Cayley complex associated with the codes $C_n$.
Given such a code, let $X$ denote its LRCC, a 2-dimensional simplicial complex.
We write $X(\ell)$ for the set of $\ell$-faces of $X$, so $X(0)$ is the set of vertices, $X(1)$ the edges, and $X(2)$ the 2-faces.
For any vertex $g\in X(0)$, the local view $X_g \subseteq X(2)$ consists of the 2-faces that contain $g$.
The LRCC construction equips each local view $X_g$ with a small code $C_g$ obtained by restricting the global code to the coordinates indexed by $X_g$.
With this terminology in place, the local testability guarantee established in \cite{dinur2022locally} can be stated as follows:
\begin{lemma}[Theorem 4.5 in \cite{dinur2022locally}]
    Let $C_n$ be an arbitrary code in \Cref{lem:thm-c3-code}, and let $X$ denote its LRCC.
    Then 
    for any $f : X(2)\to \mathbb F_2$, we have
    \[
        \Pr[g\in X(0)]{f|_{X_g}\notin C_g} \ge \kappa \cdot \delta_{\mathrm H}(f, C_n).
    \]
\end{lemma}

However, the local tests on $X_g$ are systems of linear equations rather than single linear constraints, so we translate them into tests that query a constant number of variables.
\begin{lemma}\label{lem:c3-code-new-tests}
    There exists a family of LTCs $\{\tilde C_n\}$ with constant rate, detection probability, distance, variable
    degree, and query complexity $3$ whose local tests are linear equations over $\mathbb F_2$.
\end{lemma}
\begin{proof}
    The codes in \Cref{lem:thm-c3-code} are equipped with local tests that evaluate systems of linear equations.
    To prove the lemma we construct a new family of local tests that
    \begin{enumerate}
        \itemsep=0pt
        \item consists of individual linear equations (rather than systems of linear equations); and
        \item has query complexity $3$.
    \end{enumerate}
    We address the two goals steps by step.

    \noindent\textbf{Transform local tests to linear equations.}
    By the construction of the local low-density parity-check (LDPC) codes (see the proof of Lemma 5.1 in \cite{dinur2022locally}), each local code $C_g$ is of the form $C_0\otimes C_0$ for some LDPC code $C_0$.
    Hence $f|_{X_g}\in C_g$ if and only if every row and every column of $f|_{X_g}$ lies in $C_0$.
    We view $X_g$ as an $n_0\times n_0$ grid whose entries correspond to the local coordinates, so each row/column is a subset of $X(2)$.
    Let $r_{g,1},\ldots,r_{g,n_0}\subseteq X(2)$ be these rows and $r_{g,n_0+1},\ldots,r_{g,2n_0}$ be the columns.
    Denote the parity-check matrix of $C_0$ by $H$, with rows $h_1,\ldots,h_{k_0}$.

    Sample $g\in X(0)$ uniformly, then sample $i\in[2n_0]$ and $j\in[k_0]$ uniformly and independently.
    Query the coordinates of $f$ indexed by the support of $h_j$ inside the row (or column) $r_{g,i}$ and accept iff $\inner{h_j}{f|_{r_{g,i}}}=0$.
    For a fixed $g$, if $f|_{X_g}\notin C_g$ then some pair $(i,j)$ must violate this parity check, so conditioned on such a $g$ the test rejects with probability at least $1/(2n_0k_0)$.
    Averaging over the random choice of $g$ gives
    \[
        \Pr[g,i,j]{\inner{h_j}{f|_{r_{g,i}}}\ne 0}
        \ge \frac{1}{2n_0k_0}\Pr[g\in X(0)]{f|_{X_g}\notin C_g}.
    \]
    Combining this with the local testability guarantee from Theorem~4.5 of \cite{dinur2022locally} yields
    \[
        \Pr[g,i,j]{\inner{h_j}{f|_{r_{g,i}}}\ne 0}
        \ge \frac{\kappa}{2n_0k_0}\,\delta_{\mathrm H}(f,C_n).
    \]
    Each such check touches at most $q$ variables, where $q$ is the query complexity of the base low-density parity-check code $C_0$ (Lemma~5.2 of \cite{dinur2022locally}).

    For $g\in X(0)$, $i\in[2n_0]$, and $j\in[k_0]$, let $S_{g,i,j}\subseteq X(2)$ denote the set of coordinates participating in the parity check $\inner{h_j}{f|_{r_{g,i}}}=0$.
    Let $\mathcal S_1$ collect all such $S_{g,i,j}$.
    The tests in $\mathcal S_1$ therefore define a new LTC $\tilde C_n$ with local constraints
    \[
        \sum_{s\in S_{g,i,j}} x_s = 0\qquad\text{for every }S_{g,i,j}\in\mathcal S_1,
    \]
    which coincide with the local checks of the LDPC code $C_0$.
    Because $n_0$ and $k_0$ are absolute constants, the detection probability $\kappa' := \kappa/(2n_0k_0)$ is again a constant.

    The query complexity and variable degree of this test collection are also constant.
    Indeed, the LDPC code $C_0$ has query complexity $q$ and variable degree $7$ (Lemma~5.1 of \cite{dinur2022locally}).
    A variable $x$ belongs to $S_{g,i,j}$ precisely when
    \begin{enumerate}
        \itemsep=0pt
        \item $x$ is contained in the local view $X_g$ of the LRCC (there are four such $g$’s),
        \item $x$ lies in the selected row or column $r_{g,i}$ of $X_g$ (two choices), and
        \item $x$ participates in the parity check defined by $h_j$ (at most seven choices).
    \end{enumerate}
    Therefore each variable participates in at most $4\cdot 2\cdot 7 = 56$ tests.

    \noindent\textbf{Reduce the query complexity.}
    To reduce the query complexity to $3$, we apply the standard gadget reduction from $\mathrm{E}q\mathrm{LIN}2$ to $\mathrm{E}3\mathrm{LIN}2$.
    For each $S_{g,i,j}\in\mathcal S_1$, denote its indices by $s_1,\ldots,s_q$.
    The corresponding local test is the parity constraint
    \begin{equation}\label{eq:S-1}
    x_{s_1} + \cdots + x_{s_q} = 0.
    \end{equation}
    Introduce auxiliary variables $y_{g,i,j,k}$ for $2\le k\le q-2$ and replace \cref{eq:S-1} with the chain
    \begin{equation}\label{eq:S-2}
    \begin{cases}
          x_{s_1} + x_{s_2} + y_{g,i,j,2} = 0,
        \\y_{g,i,j,2} + x_{s_3} + y_{g,i,j,3} = 0,
        \\y_{g,i,j,3} + x_{s_4} + y_{g,i,j,4} = 0,
        \\\cdots
        \\y_{g,i,j,q-3} + x_{s_{q-2}} + y_{g,i,j,q-2} = 0,
        \\y_{g,i,j,q-2} + x_{s_{q-1}} + x_{s_q} = 0.
    \end{cases}
    \end{equation}
    An assignment to $(x_s)_{s\in S_{g,i,j}}$ satisfies \cref{eq:S-1} if and only if we can extend it to the auxiliary variables so that all constraints in \cref{eq:S-2} are satisfied.
    This gadget only blows up the number of tests and the block length by a constant factor, so the rate, distance, detection probability, and variable degree remain constant, while the query complexity drops to $3$.
\end{proof}

Combining with \Cref{thm:main-intro}, we immediately get the following.
\begin{proof}[Proof of \Cref{thm:e3lin2-intro}]
    Let $\tilde C_n$ be the LTCs in \Cref{lem:c3-code-new-tests}, whose tests are linear equations.
    By \Cref{lem:lin-sys}, after a fixed permutation of coordinates we may regard each $\tilde C_n$ as systematic; let $f_n$ be the corresponding encoders.
    Then by \Cref{def:CSP-LTC}, the induced CSP instances of $(\tilde C_n, f_n)$ are exactly systems of linear equations on $\mathbb F_2$.
    Because $(\tilde C_n, f_n)$ has query complexity $3$, these instances reside in E3LIN2, so $\mathsf{MaxCSP}_{\tilde C_n}$ coincides with Max E3LIN2.

    The family $(\tilde C_n, f_n)$ inherits constant rate, distance, query complexity, detection probability, 
    and variable degree. Therefore \Cref{thm:main-intro} implies that there exists $\varepsilon > 0$
    such that every $(1-\varepsilon)$-approximation algorithm for $\mathsf{MaxCSP}_{\tilde C_n}$
    has sensitivity at least
    \[
        \frac{n}{2\theta}\left(\delta\left(1-\frac{p}{p-1}\frac{\varepsilon}{r\kappa}\right)-\frac{2\varepsilon}{\kappa}\right) - \frac{1}{2\theta} =\Omega(n),
    \]
    and the expression in parentheses is positive for sufficiently small $\varepsilon$ (depending only on the code parameters).  Hence the constant hidden in the $\Omega(\cdot)$ notation depends only on the LTC parameters. This establishes the theorem.
\end{proof}

\subsection{Repetition Codes and Max Cut}\label{subsec:repetition-code}
Let $t$ be even, $A\sqcup B=[t]$ a fixed balanced partition with $|A|=|B|=t/2$, and $k\ge 1$.
Over the alphabet $\{0,1\}$ consider the \emph{signed $t$-fold repetition code}
\[
\mathrm{Rep}^{\pm}(k,t)
=\Bigl\{ y\in\{0,1\}^{tk}:\ \exists x\in\{0,1\}^k\ \text{s.t.}\ 
y_{a,i}=x_i\ \forall a\in A,\quad y_{b,i}=1-x_i\ \forall b\in B,\ \forall i\in[k]\Bigr\}.
\]
Fix any $d$-regular bipartite graph $G=(A,B,E)$ on the vertex set $A\cup B$
($|A|=|B|=t/2$) and write $\sigma(G)\in[0,1)$ for the \emph{second singular value}
of its normalized adjacency matrix (equivalently: the largest nontrivial
absolute eigenvalue of the random-walk matrix on $G$). Consider the non-adaptive
$2$-query tester $T$ that, on input $y\in\{0,1\}^{tk}$,
\begin{quote}
chooses $i\in[k]$ uniformly at random, chooses a random edge $(a,b)\in E$,
queries the two symbols $y_{a,i},y_{b,i}$, and \emph{rejects} iff $y_{a,i}=y_{b,i}$.
\end{quote}
The blocklength is $n = tk$.

We use the following spectral Poincar\'e inequality for bipartite $d$-regular graphs.
\begin{lemma}[Edge Poincar\'e]\label{lem:poincare}
    For any $f: A\cup B\to\mathbb{R}$,
    \[
    \frac{1}{|E|}\sum_{(u,v)\in E}\bigl(f(u)-f(v)\bigr)^2
    \ge  2 (1-\sigma(G)) \mathrm{Var}(f).
    \]
\end{lemma}
\begin{proof}
    Let $A_G$ be the adjacency matrix of $G$ and $P=A_G/d$.
    A standard calculation shows
    \[
    \frac{1}{|E|}\sum_{(u,v)\in E}\bigl(f(u)-f(v)\bigr)^2
    = 2\bigl(\langle f,f\rangle-\langle f,Pf\rangle\bigr).
    \]
    Decompose $f=\bar f\cdot \mathbf{1}+z$ with $\langle z,\mathbf{1}\rangle=0$.
    Since $\mathbf{1}$ is an eigenvector of $P$ with eigenvalue $1$ and all
    other singular values are at most $\sigma(G)$ in absolute value,
    $\langle f,Pf\rangle \le \bar f^{ 2} + \sigma(G)\langle z,z\rangle$.
    Thus
    $2(\langle f,f\rangle-\langle f,Pf\rangle)
    \ge 2\bigl(\langle z,z\rangle - \sigma(G)\langle z,z\rangle\bigr)
    = 2(1-\sigma(G)) \mathrm{Var}(f)$, as claimed.
\end{proof}

\begin{theorem}\label{thm:const-degree}
    For the code $\mathrm{Rep}^{\pm}(k,t)$, the following hold:
    \begin{enumerate}
        \itemsep=0pt
        \item The relative distance is $\delta = \frac{1}{k}  =  \frac{t}{n}.$
        \item The rate is $r  =  \frac{k}{n}  =  \frac{1}{t}.$
        \item For every $y\in\{0,1\}^{n}$,
        \[
        \Pr{T\ \mathrm{rejects}\ y}
        \ge 
        \kappa\cdot \delta_{\mathrm H} \bigl(y,\mathrm{Rep}^{\pm}(k,t)\bigr),
        \qquad\text{with}\qquad
        \kappa  =  1-\sigma(G).
        \]
        In particular, $\kappa$ is an absolute constant that depends only on $d$ and the expansion of $G$ (and is independent of $t$ and $k$, hence of the blocklength $n$).
        \item Each coordinate $(v,i)$, $v\in A\cup B$, participates in exactly $d$ possible queries (one per incident edge of $G$). 
        Hence the (maximum) degree is $d$.
    \end{enumerate}
\end{theorem}
\begin{proof}
    (1)--(2) As in the ordinary repetition code, the encoder $x\mapsto y$ is injective with message length $k$ and blocklength $n=tk$;
    two distinct messages differing in one coordinate produce codewords that differ on all $t$ copies of that coordinate, giving minimum (absolute) distance $t$ and relative distance $\delta=t/n=1/k$.

    \smallskip
    (3) \emph{Soundness.}  Fix a received word $y$ and a column $i\in[k]$. Define the \emph{flipped} column
    $g_i\in\{0,1\}^{A\cup B}$ by
    \[
    g_i(v) = \begin{cases}
    y_{v,i} & v\in A,\\
    1-y_{v,i} & v\in B,
    \end{cases}
    \]
    so that $g_i$ is constant iff the $i$-th column of $y$ belongs to the code constraint
    ($A$-entries equal and $B$-entries their complements). The nearest codeword in column $i$
    is obtained by setting all $t$ positions to the majority bit of $g_i$; hence the
    per-column contribution to the relative distance is
    \[
    \Delta_i  =  \min\{\text{fraction of $1$'s in }g_i,\ \text{fraction of $0$'s in }g_i\}.
    \]
    Note that $\delta_{\mathrm H}(y,\mathrm{Rep}^{\pm}(k,t))=\frac{1}{k}\sum_{i=1}^k \Delta_i$.

    The tester $T$ samples a random edge $(a,b)\in E$ and rejects iff $y_{a,i}=y_{b,i}$,
    which (by construction of $g_i$) is equivalent to rejecting iff $g_i(a)\neq g_i(b)$.
    Therefore the rejection probability conditioned on column $i$ equals the fraction of
    \emph{disagreeing} edges of $G$ under the labeling $g_i$:
    \[
    \rho_i  :=  \Pr{\text{$T$ rejects}\mid i}
    =  \frac{1}{|E|}\sum_{(u,v)\in E} \mathbf{1} \bigl[g_i(u)\neq g_i(v)\bigr].
    \]

    We now relate $\rho_i$ to $\Delta_i$ using the expansion of $G$.
    Let $\pi$ be the uniform distribution on $A\cup B$, and write
    $\langle f,h\rangle := \mathbb{E}_{v\sim \pi}[ f(v)h(v) ]$ and
    $\mathrm{Var}(f):= \langle f-\mathbb{E}[f], f-\mathbb{E}[f]\rangle$.
    Let $P$ be the random-walk operator on $G$ (normalized adjacency).

    Apply Lemma~\ref{lem:poincare} to $f=g_i\in\{0,1\}^{A\cup B}$. Then
    $\rho_i=\frac{1}{|E|}\sum_{(u,v)}\mathbf{1}[g_i(u)\neq g_i(v)]
    = \frac{1}{|E|}\sum_{(u,v)}(g_i(u)-g_i(v))^2$,
    so
    \[
    \rho_i  \ge  2(1-\sigma(G)) \mathrm{Var}(g_i).
    \]
    For a $\{0,1\}$-valued function, $\mathrm{Var}(g_i)=\theta_i(1-\theta_i)$ where $\theta_i$ is the fraction of $1$'s;
    moreover $\Delta_i=\min\{\theta_i,1-\theta_i\}\le \frac12$, hence
    $\theta_i(1-\theta_i)\ge \frac12 \Delta_i$. Therefore,
    \[
    \rho_i  \ge  (1-\sigma(G)) \Delta_i.
    \]
    Averaging over the random column $i$ yields
    \[
    \Pr{T\ \text{rejects }y}  =  \frac{1}{k}\sum_{i=1}^k \rho_i
    \ge  (1-\sigma(G)) \frac{1}{k}\sum_{i=1}^k \Delta_i
    =  (1-\sigma(G)) \Delta \bigl(y,\mathrm{Rep}^{\pm}(k,t)\bigr),
    \]
    which proves the desired soundness with $\kappa=1-\sigma(G)$.

    \smallskip
    (4) \emph{Degree.}  Each vertex of $G$ has exactly $d$ incident edges, so a coordinate $(v,i)$ appears in exactly $d$ possible queries. 
    This $d$ is fixed once we choose $G$ and does not depend on $t$ or $n$.
\end{proof}
Using any constant-degree bipartite expander family (e.g., Ramanujan graphs), $\sigma(G)\le 2\sqrt{d-1}/d$, hence $\kappa\ge 1-2\sqrt{d-1}/d$ is a positive constant.

\begin{proof}[Proof of \Cref{thm:max-cut-intro}]
    By permuting coordinates we may regard $\mathrm{Rep}^{\pm}(k,t)$ as systematic, so every instance $\mathcal I_C(m)$ produced from the code is satisfiable and encodes a bipartite Max Cut instance.
    The code parameters are $|\Sigma|=2$, relative distance $\delta = 1/k$, rate $r=1/t$, detection probability $\kappa = 1-\sigma(G)$, and variable degree $\theta=d$.

    Let $\mathcal A$ be a $(1-\varepsilon)$-approximation algorithm for Max Cut.
    Applying \Cref{thm:sensitivity-opt-message} yields neighboring messages $m\sim m'$ with
    \[
        \EMD_{\delta_{\mathrm H}}\left(\mathcal A(\mathcal I_C(m)), \mathcal A(\mathcal I_C(m'))\right)
        \ge
        \frac{1}{k}\left(1 - \frac{2t\varepsilon}{\kappa}\right) - \frac{2\varepsilon}{\kappa} - \frac1n.
    \]
    The two instances differ only on the tester constraints that query the unique message coordinate where $m$ and $m'$ disagree; at most $d$ constraints change.
    By the triangle inequality, some such constraint $c$ satisfies
    \[
        \EMD_{\delta_{\mathrm H}}\left(\mathcal A(\mathcal I_C(m)), \mathcal A(\mathcal I_C(m)-c)\right)
        \ge
        \frac{1}{dk}\left(1 - \frac{2t\varepsilon}{\kappa}\right) - \frac{2\varepsilon}{d\kappa} - \frac1{nd}.
    \]
    Multiplying by the blocklength $n=tk$ converts the normalized Hamming metric to the unnormalized one used in the sensitivity definition, so the resulting Max Cut instance witnesses worst-case sensitivity at least
    \begin{equation}\label{eq:maxcut-sens-general}
        \frac{t}{d}\left(1 - \frac{2(t+k)\varepsilon}{\kappa}\right) - \frac{1}{d}.
    \end{equation}
    The expression inside the parentheses is positive whenever $\varepsilon \le \kappa/(2(t+k))$.  Since any algorithm has sensitivity at most $n$, the range $\varepsilon < 1/n$ is already covered by the $\Omega(n)$ lower bound for bipartite $2$-coloring~\cite{varma2023average}; we therefore focus on $1/n \le \varepsilon \le \kappa/(8\sqrt{n})$.

    Fix an even integer
    \[
        t := 2\left\lfloor \frac{\kappa}{16\varepsilon} \right\rfloor,
    \]
    so that $\kappa/(8\varepsilon) - 2 \le t \le \kappa/(8\varepsilon)$, and let $k := \left\lceil n/t \right\rceil$.
    If $t$ does not divide $n$, pad with at most $t$ dummy constraints so that the blocklength becomes $tk\in [n,n+t]$; since $t \le (\kappa/8)n$ when $\varepsilon \ge 1/n$, this increases $n$ by at most a constant factor, and we continue to denote the new blocklength by $n$ (the upper bound $\varepsilon \le \kappa/(8\sqrt{n})$ may shrink by an absolute constant, which we absorb into the choice of $8$).
    For sufficiently large $n$ (say $n\ge 16$) the assumption $\varepsilon \le \kappa/(8\sqrt{n})$ gives $\varepsilon \le \kappa/32$ and $n\varepsilon^2 \le \kappa^2/64$, whence
    \[
        t\varepsilon \le \frac{\kappa}{8}
        \qquad\text{and}\qquad
        k\varepsilon
        \le \varepsilon + \frac{n\varepsilon^2}{\kappa/8 - 2\varepsilon}
        \le \frac{\kappa}{32} + \frac{\kappa^2/64}{\kappa/16}
        = \frac{9\kappa}{32}.
    \]
    Consequently $(t+k)\varepsilon \le 13\kappa/32$, so $2(t+k)\varepsilon/\kappa \le 13/16$ and the parenthetical term in~\eqref{eq:maxcut-sens-general} is at least $3/16$.
    Substituting into~\eqref{eq:maxcut-sens-general} yields
    \[
        \mathsf{Sens}(\mathcal A,\mathcal I_C(m))
        \ge \frac{3t}{16d} - \frac{1}{d}
        \ge \frac{1}{d}\left(\frac{3\kappa}{128}\cdot\frac{1}{\varepsilon} - \frac{11}{8}\right)
        = \Omega\!\left(\frac{1}{\varepsilon}\right),
    \]
    because $\varepsilon \le \kappa/(8\sqrt{n})$ tends to $0$ with growing blocklength, so the term proportional to $1/\varepsilon$ dominates the constant.
    Thus, for every sufficiently large $n$ and every $1/n \le \varepsilon \le \kappa/(8\sqrt{n})$, we obtain a satisfiable bipartite Max Cut instance on which any $(1-\varepsilon)$-approximation algorithm has sensitivity $\Omega(1/\varepsilon)$. 
    This completes the proof.
\end{proof}

\section{Maximum $k$-Coverage}\label{sec:coverage}

We show how the sensitivity lower bound for Max E3SAT implies a matching bound for the maximum $k$-coverage problem when algorithms achieve near-perfect coverage on perfectly coverable instances.

In the maximum $k$-coverage problem, an instance $I = (U,\mathcal{F})$ has a universe $U$ and a family $\mathcal{F} \subseteq 2^{U}$. 
Outputs are $k$-subfamilies $\mathcal{S} \subseteq \mathcal{F}$ encoded by their indicator vectors $\mathbbm{1}_{\mathcal{S}} \in \{0,1\}^{\mathcal{F}}$, and we measure distances with the Hamming metric $d_{\mathrm H}(\mathbbm{1}_{\mathcal{S}}, \mathbbm{1}_{\mathcal{S}'}) = \abs{\mathcal{S} \triangle \mathcal{S}'}$. 
Two instances are neighbors if they differ only in the incidence vector of a single element $e \in U$: equivalently, we delete $e$ and insert a fresh element whose membership across $\mathcal{F}$ can be chosen arbitrarily.

\subsection{Max E3SAT}\label{sec:sat}

In Max E3SAT, an instance $\Phi = (V,\{0,1\},\mathcal C)$ consists of variables $V$ and 3-CNF clauses $\mathcal C = \set{C_1, \dots, C_m}$. The output space is $\{0,1\}^{V}$ equipped with the Hamming metric $d_{\mathrm H}$. 
Combining the standard reduction from Max E3SAT to Max E3LIN2 with a swapping variant of \Cref{thm:main-intro,thm:e3lin2-intro} (obtained by adapting its proof to replace clause deletions with clause replacements) yields the following consequence.
\begin{corollary}\label{cor:e3sat-lb}
  There exist constants $\varepsilon_0 \in (0,1)$ and $c > 0$ such that for every randomized Max E3SAT algorithm $\mathcal{A}$ that, on every satisfiable instance $\Phi = (V,\{0,1\},\mathcal C = \{C_1,\ldots,C_m\})$, outputs in expectation a $(1-\varepsilon_0)$-approximate assignment, there is a satisfiable $\Phi$ with
  \begin{equation}\label{eq:e3sat-lb}
    \max_{i \in [m]}\;\max_{C' \in \mathcal{L}_i}
    \EMD_{d_{\mathrm H}}\bigl(\mathcal{A}(\Phi), \mathcal{A}(\Phi - C_i + C')\bigr)
    \ge c |V|.
  \end{equation}
  Here $\mathcal{L}_i$ denotes the set of clauses on the three variables that appear in $C_i$, and $\Phi - C_i + C'$ is the instance obtained from $\Phi$ by replacing $C_i$ with $C'$.
\end{corollary}

\subsection{Block Gadget Reduction}
Fix integers $n = \abs{V}$ and $k \le n$.
Without loss of generality, we assume that $n$ is divisible by $k$.
Partition $V$ into $k$ disjoint blocks $V = V_1 \sqcup \cdots \sqcup V_k$ with $\abs{V_i} = n/k$. 
For each $i \in [k]$ and $\alpha \in \{0,1\}^{V_i}$ create a set $S_{i,\alpha}$. For every clause $C_j \in C$ create an element $e_j$, and include $e_j \in S_{i,\alpha}$ whenever the restriction $\alpha$ satisfies a literal of $C_j$ whose variable lies in $V_i$. 
Let $\mathcal{F} = \set{S_{i,\alpha} \mid i \in [k], \alpha \in \{0,1\}^{V_i}}$.

Let $m = \abs{C}$.
Without loss of generality, we assume that $m$ is divisible by $k$.
We choose $M = m/k$. 
For each block $i$, add a multiset $G_i$ of $M$ new guard elements. Every set $S_{i,\alpha}$ contains the guards $G_i$, while no set $S_{i',\alpha'}$ with $i' \neq i$ contains them. 
Write $U_{\mathrm{cl}} = \set{e_1, \dots, e_m}$, $U_{\mathrm{gd}} = \bigsqcup_{i=1}^k G_i$, and define $U^+ = U_{\mathrm{cl}} \cup U_{\mathrm{gd}}$. 
The resulting maximum $k$-coverage problem instance is denoted by $R_k^+(\Phi) = (U^+, \mathcal{F})$.

\begin{lemma}\label{lem:cov=sat}
For any choice of exactly one set per block, say $\set{S_{i,\alpha_i}}_{i=1}^k$, let $x \in \{0,1\}^{V}$ be the concatenation of the block assignments $\alpha_1, \dots, \alpha_k$. Then
\[
  \left| \left( \bigcup_{i=1}^k S_{i,\alpha_i} \right) \cap U_{\mathrm{cl}} \right|
  = \left|\set{C_j \in C \mid x \text{ satisfies } C_j}\right|.
\]
\end{lemma}

\begin{lemma}[Neighbor preservation]\label{lem:neighbor}
If $\Phi'$ is obtained from $\Phi$ by modifying one clause $C_t$, then $R_k^+(\Phi')$ is obtained from $R_k^+(\Phi)$ by reprogramming the single element $e_t$ while leaving all guards untouched.
\end{lemma}

\begin{lemma}[Near-perfect coverage forces one choice per block]\label{lem:one-per-block}
Let $I^+ = R_k^+(\Phi)$ and $\varepsilon \in (0,1)$. Suppose a $k$-subfamily $\mathcal{S}$ covers at least $(1-\varepsilon) \abs{U^+}$ elements. If $M > \varepsilon \abs{U^+}$, then $\mathcal{S}$ contains at least one set from each block. As $\abs{\mathcal{S}} = k$, it in fact contains exactly one from every block and therefore covers all guards $U_{\mathrm{gd}}$.
\end{lemma}

\begin{proof}
If block $i$ contributes no set, all $M$ guards in $G_i$ remain uncovered, contradicting $\abs{\bigcup_{S \in \mathcal{S}} S} \ge (1-\varepsilon) \abs{U^+}$ when $M > \varepsilon \abs{U^+}$. With $k$ blocks and budget $k$, covering each block at least once is equivalent to selecting exactly one per block.
\end{proof}

\paragraph{A canonical decoder.}
Whenever $\mathcal{S}$ contains exactly one set from each block we write it as $\mathcal{S} = \set{S_{i,\alpha_i}}_{i=1}^k$ and define $D(\mathcal{S}) \in \{0,1\}^{V}$ to be the concatenation of the block assignments.

\begin{lemma}[Decoder Lipschitzness]\label{lem:decoder-lip}
Let $\mathcal{S}, \mathcal{S}'$ be $k$-subfamilies that each contain exactly one set per block. Then
\[
  d_{\mathrm H}\bigl(D(\mathcal{S}), D(\mathcal{S}')\bigr)
  \le \frac{n}{2k} \cdot d_{\mathrm H}(\mathbbm{1}_{\mathcal{S}}, \mathbbm{1}_{\mathcal{S}'}).
\]
\end{lemma}

\begin{proof}
Changing the choice within one block removes one set and adds another, increasing the symmetric difference by two. Each affected block modifies at most $n/k$ assignment bits, so the claimed Lipschitz bound follows.
\end{proof}

\begin{lemma}\label{lem:approx-transfer}
Assume $\mathcal{S}$ covers at least $(1-\varepsilon) \abs{U^+}$ elements and, by \Cref{lem:one-per-block}, contains exactly one set per block. Then the decoded assignment $x = D(\mathcal{S})$ satisfies at least $(1-2\varepsilon) m$ clauses. In particular, this conclusion holds whenever $M \le m/k$, because then $\abs{U^+} = m + Mk \le 2m$ and all guards are covered.
\end{lemma}

\begin{proof}
All $Mk$ guards are covered. Hence the number of covered clause elements is at least
$(1-\varepsilon) \abs{U^+} - Mk = (1-\varepsilon)(m + Mk) - Mk = (1-\varepsilon) m - \varepsilon Mk \ge (1-2\varepsilon) m$
whenever $Mk \le m$. By \Cref{lem:cov=sat}, this equals the number of satisfied clauses.
\end{proof}

\subsection{Sensitivity Transfer}
We now relate algorithms for the maximum $k$-coverage problem to algorithms for Max E3SAT.

\begin{theorem}[Restatement of Theorem~\ref*{thm:coverage-intro}]\label{thm:coverage-transfer}
There exists $\varepsilon_0 > 0$ such that, for every $\varepsilon \in (0, \varepsilon_0)$ and every $k < 1/(2\varepsilon)$, any randomized $(1-\varepsilon)$-approximation algorithm $\mathcal{B}$ for the maximum $k$-coverage problem (in expectation) has an instance $I^+$ with $\OPT(I^+) = \abs{U^+}$ on which
\[
  \mathrm{Sens}(\mathcal{B},I^+) = \Omega(k),
\]
with an implicit constant that is universal.
\end{theorem}

\begin{proof}
Let $\varepsilon_0 \in (0,1)$ and $c > 0$ be the constants supplied by \Cref{cor:e3sat-lb}.
Fix $\varepsilon \in (0, \varepsilon_0)$, $k < 1/(2\varepsilon)$, and a randomized $(1-\varepsilon)$-approximation algorithm $\mathcal{B}$.

\textit{Step 1: Build an E3SAT algorithm from $\mathcal{B}$.}
Define $\mathcal{A}_{\mathrm{sat}} := D \circ \mathcal{B} \circ R_k^+$ and set $M := m/k$ in the reduction $R_k^+$.

\smallskip
\textit{Step 2: $\mathcal{A}_{\mathrm{sat}}$ is $(1-\varepsilon_0)$-accurate on satisfiable inputs.}
When $\Phi$ is satisfiable, choosing the set whose block restriction matches a satisfying assignment yields perfect coverage of $R_k^+(\Phi)$. Because $M = m/k$ and $k < 1/(2\varepsilon)$, we have
$M > \varepsilon \abs{U^+}$ (indeed $\abs{U^+} = m + Mk \le 2m$ and $M \ge m/k > 2\varepsilon m$). Therefore \Cref{lem:one-per-block} ensures that $\mathcal{B}(R_k^+(\Phi))$ selects exactly one set per block, and \Cref{lem:approx-transfer} gives that
$D \circ \mathcal{B}(R_k^+(\Phi))$ satisfies at least $(1-2\varepsilon) m \ge (1-\varepsilon_0) m$ clauses in expectation.

\smallskip
\textit{Step 3: Invoke the Max E3SAT sensitivity bound.}
We choose $\varepsilon_0$ to be small enough.
Then by \eqref{eq:e3sat-lb}, there exist a satisfiable $\Phi$ and a clause $C_t'$ on the same variable triple as some $C_t \in \mathcal C$ such that the instance $\Phi' := \Phi - C_t + C_t'$ satisfies
\[
  \EMD_{d_{\mathrm H}}(\mathcal{A}_{\mathrm{sat}}(\Phi), \mathcal{A}_{\mathrm{sat}}(\Phi')) \ge c n,
\]
where $n = \abs{V}$. Let $I^+ = R_k^+(\Phi)$ and $I^{+\prime} = R_k^+(\Phi')$. By \Cref{lem:neighbor}, $I^{+\prime}$ is obtained from $I^+$ by reprogramming the single clause element $e_t$, so it is a valid neighbor in the maximum $k$-coverage problem.

\smallskip
\textit{Step 4: EMD contraction through the decoder.}
Let $\mathcal{F}^{(k)}$ denote the set of $k$-subfamilies of $\mathcal{F}$. Conditioning on the event that $\mathcal{B}$ selects exactly one set per block (which holds for both $I^+$ and $I^{+\prime}$ by the guard argument above), \Cref{lem:decoder-lip} shows that the decoder $D$ is $n/(2k)$-Lipschitz from $(\mathcal{F}^{(k)}, d_{\mathrm H})$ to $(\{0,1\}^{V}, d_{\mathrm H})$. Let $D_{\ast}\mu$ denote the pushforward of a distribution $\mu$ under $D$. The Lipschitz contraction property for earthmover distance therefore implies
\[
  \EMD_{d_{\mathrm H}}(D_{\ast} \mathcal{B}(I^+), D_{\ast} \mathcal{B}(I^{+\prime}))
  \le \frac{n}{2k} \cdot \EMD_{d_{\mathrm H}}(\mathcal{B}(I^+), \mathcal{B}(I^{+\prime})).
\]
Combining with Step~3 and rearranging gives $\EMD_{d_{\mathrm H}}(\mathcal{B}(I^+), \mathcal{B}(I^{+\prime})) \ge \Omega(k)$, which is precisely $\mathrm{Sens}(\mathcal{B},I^+) = \Omega(k)$. Because $R_k^+(\Phi)$ is perfectly coverable, this concludes the promised instance.
\end{proof}

\section{A Boolean-Function View of Randomized Algorithms for Satisfiable Max E3SAT}
\label{sec:sat-cumulative-influence}

In this section we recast the instability results from \Cref{sec:tv-distance,sec:sat} in the
language of Boolean functions. We study randomized algorithms that, given a satisfiable Max E3SAT
formula, output an assignment satisfying almost all clauses in expectation. Our hard instances are
not arbitrary unrelated formulas: they form an explicit local subcube inside the ordinary SAT input
space, indexed by messages $m \in \{0,1\}^k$. All formulas in the family share the same variables
and the same ordered clause skeleton, and flipping one bit of $m$ changes only $O(1)$ clause signs.
Once the algorithm's random seed is fixed, each output bit becomes an ordinary Boolean function of $m$. Its influence is the sum, over
all message coordinates, of the probability that flipping that coordinate changes the output bit.
The main conclusion is that, on this explicit satisfiable family, these output bits cannot all have
small influence: averaged over the random seed, the sum of their influences is necessarily large.
This viewpoint is what later yields the junta, pseudo-junta, decision-tree, and bounded-depth
circuit corollaries.

We say that a randomized algorithm $A$ is \emph{one-sided} $(1-\eta)$-approximate for Max E3SAT
if $A$ is defined on every $3$-CNF formula and, for every satisfiable input formula $\Phi$,
\[
\mathbb{E}_{a \sim A(\Phi)}
\bigl[
\#\text{ clauses of }\Phi\text{ satisfied by }a
\bigr]
\ge (1-\eta)|\Phi|.
\]
Equivalently, if $R$ denotes the internal random seed of $A$, then
\[
\mathbb{E}_{R}
\bigl[
\#\text{ clauses of }\Phi\text{ satisfied by }A(\Phi;R)
\bigr]
\ge (1-\eta)|\Phi|
\]
for every satisfiable $\Phi$. No guarantee is required on unsatisfiable inputs.
When $A$ happens to be deterministic, the seed $R$ is trivial and all of the statements below
reduce to the deterministic version.

\paragraph{The standard formula encoding.}
Fix integers $n$ and $M$, and regard a $3$-CNF formula on variable set $[n]$ with exactly
$M$ clauses as an ordered list
\[
\Phi = (C_1,\dots,C_M).
\]
Write each clause in the canonical order of its three variable indices, and encode each literal
by the binary representation of its variable index in $[n]$ together with one sign bit indicating
whether it is negated. Let
\[
\operatorname{Enc}_{n,M}(\Phi) \in \{0,1\}^{N},
\qquad
N := M\bigl(3\lceil \log_2 n\rceil + 3\bigr),
\]
denote the concatenation of the $M$ clause encodings. Throughout this section, when we write
$A(\Phi)$ for such a randomized algorithm $A$, we mean the output distribution
$A(\operatorname{Enc}_{n,M}(\Phi))$ under this standard clause-by-clause encoding. Thus the ambient
input space is the completely ordinary bit representation of $3$-CNF formulas; the parameter
$m \in \{0,1\}^k$ used below merely indexes an explicit hard subfamily inside that ambient space.

\subsection{From Instability to Influence}

We begin with the randomized E3LIN2 instability statement supplied by the LTC framework.

\begin{lemma}
\label{lem:e3lin2-endpoint-instability}
There exist constants $\eta_{\mathrm{lin}}, c_{\mathrm{lin}} > 0$ and an explicit family
$\{I(m)\}_{m \in \{0,1\}^k}$ of satisfiable bounded-occurrence Max E3LIN2 instances on
$n = \Theta(k)$ variables such that every randomized algorithm $B$ satisfying
\[
\mathbb{E}_{R}
\bigl[
\#\text{ equations of }I(m)\text{ satisfied by }B(I(m);R)
\bigr]
\ge (1-\eta_{\mathrm{lin}})|I(m)|
\qquad
\forall m \in \{0,1\}^k
\]
obeys
\[
\mathbb{E}_{m \in \{0,1\}^k}
\mathbb{E}_{i \in [k]}
\Bigl[
 \operatorname{EMD}_{d_{\mathrm H}}\bigl(B(I(m)), B(I(m^{\oplus i}))\bigr)
\Bigr]
\ge c_{\mathrm{lin}}\, n,
\]
where $m^{\oplus i}$ denotes $m$ with its $i$th bit flipped.
\end{lemma}

\begin{proof}
This is exactly the averaged earth-mover lower bound obtained from
\Cref{lem:emd-expectation} using the binary LRCC family from
\Cref{subsec:lrcc-max-e3lin2}. Since the rate, distance, detection probability, and
variable degree are all absolute constants, the right-hand side of
\Cref{lem:emd-expectation} is bounded below by a positive constant for sufficiently small
$\eta_{\mathrm{lin}}$. Multiplying by the block length converts the normalized Hamming metric
to $d_{\mathrm H}$, and the resulting Max E3LIN2 instances have $\Theta(n)$ output coordinates.
Hence the averaged instability is $\Omega(n)$.
\end{proof}

We now pass from E3LIN2 to E3SAT using the standard four-clause gadget. Because the E3LIN2 family
already comes with a fixed ordered list of equation scopes, this step preserves a common template:
varying $m$ changes only the right-hand-side bits of $O(1)$ equations, hence only the sign patterns
in $O(1)$ gadgets.
For each equation
\[
x \oplus y \oplus z = b,
\]
let $G_b(x,y,z)$ denote the standard $4$-clause gadget on the same variable triple
with the property that any assignment satisfies all four clauses iff it satisfies
the equation, and otherwise it satisfies exactly three clauses.

\begin{theorem}[One-sided endpoint instability for ordinarily encoded satisfiable Max E3SAT]
\label{thm:sat-one-sided-cumulative-influence}
There exist constants $\eta, c, \Delta > 0$, an integer $L = O(1)$, and an explicit family
$\{\Phi(m)\}_{m \in \{0,1\}^k}$ of satisfiable bounded-occurrence $3$-CNF formulas on the
same variable set $[n]$ with the same number $M = O(n)$ of clauses, where $n = \Theta(k)$,
such that the following holds.

Let $A$ be any randomized algorithm that is one-sided $(1-\eta)$-approximate for Max E3SAT.
Then:

\begin{enumerate}
\item \label{item:sat-endpoint-version}\textbf{Endpoint instability.}
\[
\mathbb{E}_{m \in \{0,1\}^k}
\mathbb{E}_{i \in [k]}
\Bigl[
 \operatorname{EMD}_{d_{\mathrm H}}\bigl(A(\Phi(m)), A(\Phi(m^{\oplus i}))\bigr)
\Bigr]
\ge c\, n.
\]

\item \label{item:sat-single-clause-path-version}\textbf{Single-clause path version.}
For every $m \in \{0,1\}^k$ and $i \in [k]$, there exists a sequence
\[
\Phi(m) = \Psi^{m,i}_0,\;
\Psi^{m,i}_1,\;
\dots,\;
\Psi^{m,i}_L = \Phi(m^{\oplus i})
\]
such that each consecutive pair $\Psi^{m,i}_{t-1}, \Psi^{m,i}_t$ differs by replacing
a single clause by another clause on the same variable triple, and
\[
\mathbb{E}_{m \in \{0,1\}^k}
\mathbb{E}_{i \in [k]}
\Biggl[
\sum_{t=1}^L
 \operatorname{EMD}_{d_{\mathrm H}}\bigl(A(\Psi^{m,i}_{t-1}), A(\Psi^{m,i}_t)\bigr)
\Biggr]
\ge c\, n.
\]
\end{enumerate}

Moreover, the hard family is local inside the ordinary SAT input space: for every
$m \in \{0,1\}^k$ and $i \in [k]$,
\[
d_{\mathrm H}\bigl(\operatorname{Enc}_{n,M}(\Phi(m)),\operatorname{Enc}_{n,M}(\Phi(m^{\oplus i}))\bigr)
\le \Delta,
\]
and each adjacent pair $\Psi^{m,i}_{t-1},\Psi^{m,i}_{t}$ in
item~\ref{item:sat-single-clause-path-version}
differs in at most three input bits under $\operatorname{Enc}_{n,M}$.
In particular, the second statement is \emph{one-sided}: only the endpoint formulas
$\Phi(m)$ and $\Phi(m^{\oplus i})$ are required to be satisfiable; the intermediate
formulas $\Psi^{m,i}_t$ may be arbitrary.
\end{theorem}

\begin{proof}
Let $\{I(m)\}_{m \in \{0,1\}^k}$ be the satisfiable Max E3LIN2 family from
\Cref{lem:e3lin2-endpoint-instability}. For each $m$, define $\Phi(m)$ by
replacing every equation of $I(m)$ with the corresponding four-clause gadget.
Order the clauses of $\Phi(m)$ by first fixing the order of the E3LIN2 constraints and then,
inside each gadget, fixing the order of its four clauses.
The key point is that this ordered list of equation scopes is independent of $m$; only the
right-hand-side bits vary with the message.
Because the reduction is local and each variable appears in only constantly many
equations in the E3LIN2 family, the resulting formulas are explicit, satisfiable,
and bounded-occurrence, with $M = O(n)$ clauses.

Set $\eta := \eta_{\mathrm{lin}}/4$, and let $A$ be a randomized one-sided $(1-\eta)$-approximate
algorithm for Max E3SAT. Define
\[
B(I(m)) := A(\Phi(m)).
\]
Let $Y$ be a random assignment drawn from $A(\Phi(m))$, and let $s(Y)$ denote the number of
E3LIN2 equations of $I(m)$ satisfied by $Y$. Since each satisfied equation contributes $4$
satisfied clauses and each unsatisfied equation contributes exactly $3$, we have
\[
\#\text{ clauses of }\Phi(m)\text{ satisfied by }Y
=
3|I(m)| + s(Y).
\]
Therefore
\[
\mathbb{E}[s(Y)]
\ge
(1-\eta)\,4|I(m)| - 3|I(m)|
=
(1-4\eta)|I(m)|
=
(1-\eta_{\mathrm{lin}})|I(m)|.
\]
Thus $B$ satisfies the hypothesis of \Cref{lem:e3lin2-endpoint-instability},
and item~\ref{item:sat-endpoint-version} follows immediately with $c := c_{\mathrm{lin}}$.

It remains to justify the locality claims in the standard SAT input encoding.
By construction of the LTC-induced E3LIN2 family,
the two instances $I(m)$ and $I(m^{\oplus i})$ differ only on those local tests
that read the $i$th message coordinate, hence on at most $\theta$ equations,
where $\theta = O(1)$ is the variable degree of the tester.
For every fixed equation slot, the three variables appearing in that equation are
independent of $m$; only the right-hand-side bit may change with $m$.
Consequently, when such an equation is replaced by the corresponding E3SAT gadget,
all variable-index bits in the standard clause-listing encoding remain fixed,
and only the sign bits of the literals may change.
Each changed equation slot affects four clauses, each clause contributes three sign bits,
and therefore
\[
d_{\mathrm H}\bigl(\operatorname{Enc}_{n,M}(\Phi(m)),\operatorname{Enc}_{n,M}(\Phi(m^{\oplus i}))\bigr)
\le \Delta := 12\theta.
\]
This proves the bounded-difference claim for the endpoints.

For item~\ref{item:sat-single-clause-path-version}, fix $m$ and $i$.
Changing one E3LIN2 equation changes one four-clause gadget, and that gadget change can be
implemented by four single-clause replacements. Since the source and target clauses in each such
replacement use the same variable triple, their standard encodings differ only in the three sign bits
of that clause. Therefore there is a sequence
\[
\Phi(m)=\Psi^{m,i}_0,\Psi^{m,i}_1,\dots,\Psi^{m,i}_{L}=\Phi(m^{\oplus i})
\]
of length at most $L:=4\theta$ in which each step replaces a single clause
by another clause on the same variable triple. (If a shorter path exists, pad it
with repeated endpoint formulas.) By the triangle inequality for $\operatorname{EMD}_{d_{\mathrm H}}$,
\[
\operatorname{EMD}_{d_{\mathrm H}}\bigl(A(\Phi(m)), A(\Phi(m^{\oplus i}))\bigr)
\le
\sum_{t=1}^{L}
 \operatorname{EMD}_{d_{\mathrm H}}\bigl(A(\Psi^{m,i}_{t-1}), A(\Psi^{m,i}_t)\bigr).
\]
Averaging over $m$ and $i$ and using item~\ref{item:sat-endpoint-version} completes the proof.
\end{proof}

\paragraph{The hard embedding is itself simple.}
Define the explicit input embedding
\[
x(m) := \operatorname{Enc}_{n,M}(\Phi(m)) \in \{0,1\}^{N}.
\]
Each output bit of $x$ is either a fixed variable-index bit or a sign bit of a clause in one gadget,
and the sign pattern of that gadget is determined by a single local E3LIN2 test. Since the tester
query complexity is constant, each output bit of $x$ depends on only $O(1)$ message coordinates.
Thus $x$ is computable by a linear-size $NC^0$ circuit. This is important for the randomized
circuit corollary below: the hardness is not coming from a complicated front end, but from the
randomized algorithm itself.

Because the relevant input family and output space are finite, any such randomized algorithm can be
realized on a common finite probability space. We therefore fix a random seed $R$ and a realization
$A(\cdot;R)$ whose law is $A(\cdot)$ on every input. For the rest of the section write
\[
F_R(m) := A(x(m);R) = \bigl(f_{1,R}(m),\dots,f_{n,R}(m)\bigr) \in \{0,1\}^n.
\]
For each fixed seed $r$ and output coordinate $v$, the slice
$f_{v,r} : \{0,1\}^k \to \{0,1\}$ is an ordinary Boolean function.
For a Boolean function $f : \{0,1\}^k \to \{0,1\}$, write
\[
\operatorname{Inf}(f)
:=
\sum_{i=1}^k
\Pr[m \in \{0,1\}^k]{f(m)\neq f(m^{\oplus i})}.
\]
We call
\[
\operatorname{SCInf}(A)
:=
\mathbb{E}_{R}
\Bigl[
\sum_{v=1}^n \operatorname{Inf}(f_{v,R})
\Bigr]
\]
the \emph{seed-averaged cumulative influence} of the pullback of $A$ to the hard family.
When $A$ is deterministic, this is exactly the ordinary cumulative influence from the deterministic
version of the section.

\paragraph{Why the pullback is the right object.}
The theorem above is already stated in the standard formula-input model.
We pass to the cube coordinates $m \in \{0,1\}^k$ because the classical structural results we want
to invoke, i.e., junta theorems, pseudo-junta theorems, and average-sensitivity upper bounds for
bounded-depth circuits and formulas, are formulated for Boolean functions on product spaces.
Moreover, by \Cref{thm:sat-one-sided-cumulative-influence}, adjacent points of the hard cube map
to formulas whose standard encodings differ in only $O(1)$ bits, so the pullback does not create
artificial hardness: it is simply a convenient coordinate system for an explicit, genuinely local
subfamily of ordinary SAT inputs. The role of the random seed is only to turn the randomized
algorithm into deterministic slices to which these classical Boolean-function bounds apply.

\begin{corollary}[Seed-averaged cumulative influence]
\label{cor:sat-deterministic-cumulative-influence}
Under the hypotheses of \Cref{thm:sat-one-sided-cumulative-influence}, every randomized
algorithm $A$ that is one-sided $(1-\eta)$-approximate for Max E3SAT satisfies
\[
\operatorname{SCInf}(A)
=
\mathbb{E}_{R}
\Bigl[
\sum_{v=1}^n \operatorname{Inf}(f_{v,R})
\Bigr]
= \Omega(nk).
\]
Moreover, for the single-clause paths from
\Cref{thm:sat-one-sided-cumulative-influence},
\[
\sum_{v=1}^n \sum_{i=1}^k \sum_{t=1}^{L}
\Pr[m,R]{
A(\Psi^{m,i}_{t-1};R)_v \neq A(\Psi^{m,i}_t;R)_v
}
= \Omega(nk).
\]
\end{corollary}

\begin{proof}
For the endpoint statement, couple the two output distributions
$A(x(m))$ and $A(x(m^{\oplus i}))$ by using the same random seed $R$ on both inputs.
This is a valid coupling, so
\[
\operatorname{EMD}_{d_{\mathrm H}}\bigl(A(x(m)), A(x(m^{\oplus i}))\bigr)
\le
\mathbb{E}_{R}
\Bigl[
 d_{\mathrm H}\bigl(F_R(m), F_R(m^{\oplus i})\bigr)
\Bigr].
\]
Averaging over $m$ and $i$ gives
\[
\mathbb{E}_{m,i}
\Bigl[
 \operatorname{EMD}_{d_{\mathrm H}}\bigl(A(x(m)), A(x(m^{\oplus i}))\bigr)
\Bigr]
\le
\mathbb{E}_{m,i,R}
\Bigl[
 d_{\mathrm H}\bigl(F_R(m), F_R(m^{\oplus i})\bigr)
\Bigr].
\]
Expanding Hamming distance coordinatewise and then unfolding the definition of total influence,
\[
\mathbb{E}_{m,i,R}
\Bigl[
 d_{\mathrm H}\bigl(F_R(m), F_R(m^{\oplus i})\bigr)
\Bigr]
=
\frac{1}{k}
\mathbb{E}_{R}
\Bigl[
\sum_{v=1}^n \operatorname{Inf}(f_{v,R})
\Bigr].
\]
Applying item~\ref{item:sat-endpoint-version} of
\Cref{thm:sat-one-sided-cumulative-influence} yields
\[
\frac{1}{k}
\mathbb{E}_{R}
\Bigl[
\sum_{v=1}^n \operatorname{Inf}(f_{v,R})
\Bigr]
= \Omega(n),
\]
which is equivalent to the displayed lower bound.

The path version is identical: for each $m,i,t$,
\[
\operatorname{EMD}_{d_{\mathrm H}}\bigl(A(\Psi^{m,i}_{t-1}), A(\Psi^{m,i}_{t})\bigr)
\le
\mathbb{E}_{R}
\Bigl[
 d_{\mathrm H}\bigl(A(\Psi^{m,i}_{t-1};R), A(\Psi^{m,i}_{t};R)\bigr)
\Bigr],
\]
and summing over $t$, averaging over $m$ and $i$, and expanding coordinatewise proves the second
claim.
\end{proof}

\begin{corollary}[Many highly influential output coordinates on average over the seed]
\label{cor:sat-many-influential-coordinates}
There exist constants $\alpha,\gamma>0$ such that
\[
\mathbb{E}_{R}
\Bigl[
\bigl|\{v \in [n] : \operatorname{Inf}(f_{v,R}) \ge \alpha k\}\bigr|
\Bigr]
\ge \gamma n.
\]
In particular, for at least one choice of the internal randomness $r$, at least $\gamma n$
output coordinates satisfy
\[
\operatorname{Inf}(f_{v,r}) \ge \alpha k.
\]
Hence on at least one seed slice, linearly many output coordinates depend on $\Omega(k)$
message coordinates and have deterministic decision-tree depth $\Omega(k)$.
\end{corollary}

\begin{proof}
By \Cref{cor:sat-deterministic-cumulative-influence}, there exists a constant $c_0>0$ such that
\[
\mathbb{E}_{R}
\Bigl[
\sum_{v=1}^n \operatorname{Inf}(f_{v,R})
\Bigr]
\ge c_0 nk.
\]
Set $\alpha := c_0/2$, and for each seed $r$ define
\[
H_r := \{v \in [n] : \operatorname{Inf}(f_{v,r}) \ge \alpha k\}.
\]
Since $\operatorname{Inf}(f_{v,r}) \le k$ for every $v$ and $r$,
\[
\sum_{v=1}^n \operatorname{Inf}(f_{v,r})
\le
|H_r|\,k + (n-|H_r|)\alpha k.
\]
Taking expectations over $R$ and rearranging shows that
\[
\mathbb{E}_{R}[|H_R|] \ge \gamma n
\]
for some constant $\gamma>0$.

If a Boolean function depends on at most $T$ input coordinates, then its total influence
is at most $T$. Likewise, the deterministic decision-tree depth $D(f)$ satisfies
$\operatorname{Inf}(f)\le D(f)$. Therefore every coordinate counted by $H_r$ depends on
$\Omega(k)$ message bits and has decision-tree depth $\Omega(k)$.
\end{proof}

\subsection{A Meta-Corollary for Low-Influence Representation Classes}

For a class $\mathcal{G}_k$ of Boolean functions on $\{0,1\}^k$, define
\[
\operatorname{dist}(f,\mathcal{G}_k)
:=
\inf_{g \in \mathcal{G}_k}
\Pr[m \in \{0,1\}^k]{f(m)\neq g(m)}.
\]

\begin{corollary}[A meta-corollary for randomized low-influence classes]
\label{cor:sat-low-influence-meta}
Let $\alpha,\gamma>0$ be the constants from \Cref{cor:sat-many-influential-coordinates}.
Let $\mathcal{G}_k$ be any class of Boolean functions on $\{0,1\}^k$ satisfying
\[
\sup_{g \in \mathcal{G}_k} \operatorname{Inf}(g) \le \beta(k).
\]
If $\beta(k) \le \alpha k/2$, then
\[
\mathbb{E}_{R}
\Bigl[
\bigl|\{v \in [n] : \operatorname{dist}(f_{v,R},\mathcal{G}_k) \ge \alpha/4\}\bigr|
\Bigr]
\ge \gamma n.
\]
\end{corollary}

\begin{proof}
For any two Boolean functions $f,g : \{0,1\}^k \to \{0,1\}$,
\[
\bigl|\operatorname{Inf}(f)-\operatorname{Inf}(g)\bigr|
\le 2k\Pr[m \in \{0,1\}^k]{f(m)\neq g(m)}.
\]
Indeed, for each $i \in [k]$,
\[
\Bigl|\Pr[m]{f(m)\neq f(m^{\oplus i})}-\Pr[m]{g(m)\neq g(m^{\oplus i})}\Bigr|
\le 2\Pr[m]{f(m)\neq g(m)},
\]
and summing over $i$ proves the claim.

Now fix a seed $r$ and let $v$ be any coordinate with $\operatorname{Inf}(f_{v,r}) \ge \alpha k$.
For every $g \in \mathcal{G}_k$ we then have
\[
\Pr[m]{f_{v,r}(m)\neq g(m)}
\ge
\frac{\operatorname{Inf}(f_{v,r})-\operatorname{Inf}(g)}{2k}
\ge
\frac{\alpha k-\beta(k)}{2k}
\ge \alpha/4.
\]
Thus
\[
\bigl|\{v \in [n] : \operatorname{dist}(f_{v,r},\mathcal{G}_k) \ge \alpha/4\}\bigr|
\ge
\bigl|\{v \in [n] : \operatorname{Inf}(f_{v,r}) \ge \alpha k\}\bigr|.
\]
Taking expectations over $R$ and applying
\Cref{cor:sat-many-influential-coordinates} proves the displayed lower bound.
The final sentence is immediate.
\end{proof}

The strength of \Cref{cor:sat-low-influence-meta} is that it turns
\Cref{cor:sat-deterministic-cumulative-influence} into a black-box lower-bound principle:\ \emph{any} coordinatewise representation class whose seedwise realizations have sublinear total
influence is automatically ruled out on linearly many output coordinates on average over the seed.
The next three subsections record concrete instantiations.

\subsection{Anti-Junta and Anti-Pseudo-Junta Consequences}

Recall that a Boolean function on $\{0,1\}^k$ is a \emph{$K$-junta} if it depends on at most $K$
input coordinates. We also use Hatami's notion of a \emph{$K$-pseudo-junta}~\cite{Hatami12},
which generalizes juntas by allowing the relevant coordinates to vary with the input in a
structured way. We will not need the full formal definition here; the only property used below is
that every $K$-pseudo-junta has total influence at most $2K$.

\begin{corollary}[On average over the seed, linearly many outputs are far from juntas and pseudo-juntas]
\label{cor:sat-anti-junta}
There exist constants $\gamma,\delta>0$ such that the following holds for all sufficiently large $k$.

\begin{enumerate}
\item For every $K \le \gamma k$, one has
\[
\mathbb{E}_{R}
\Bigl[
\bigl|\{v \in [n] : \operatorname{dist}(f_{v,R},\{\text{$K$-juntas}\}) \ge \delta\}\bigr|
\Bigr]
\ge \delta n.
\]

\item For every $K \le \gamma k$, one has
\[
\mathbb{E}_{R}
\Bigl[
\bigl|\{v \in [n] : \operatorname{dist}(f_{v,R},\{\text{$K$-pseudo-juntas}\}) \ge \delta\}\bigr|
\Bigr]
\ge \delta n.
\]
\end{enumerate}
\end{corollary}

\begin{proof}
Let $\alpha_0,\gamma_0$ be the constants from
\Cref{cor:sat-many-influential-coordinates}.
Every $K$-junta satisfies $\operatorname{Inf}(g) \le K$~\cite{Friedgut98}, and every $K$-pseudo-junta satisfies
$\operatorname{Inf}(g) \le 2K$ by Hatami's proposition on pseudo-juntas~\cite{Hatami12}.
Thus both claims follow from \Cref{cor:sat-low-influence-meta} by taking
$\delta := \min\{\alpha_0/4,\gamma_0\}$ and shrinking the displayed constant $\gamma$ by
an absolute factor if needed so that both $K$ and $2K$ are at most $\alpha_0 k/2$.
\end{proof}

The point of \Cref{cor:sat-anti-junta} is that the randomized task of finding near-satisfying
assignments for an explicit satisfiable
Max E3SAT family forces averaged non-junta structure: linearly many output coordinates
stay a constant distance away from every small junta and every small pseudo-junta on
average over the internal randomness.

\subsection{Decision Trees and Local Dependence}

For a Boolean function $f : \{0,1\}^k \to \{0,1\}$, let $D(f)$ denote its deterministic
decision-tree depth, that is, the minimum worst-case number of input bits that must be queried in
order to determine $f(m)$.

\begin{corollary}[On average over the seed, linearly many outputs have linear decision-tree depth]
\label{cor:sat-dt-depth}
There exist constants $\alpha,\gamma>0$ such that
\[
\mathbb{E}_{R}
\Bigl[
\bigl|\{v \in [n] : D(f_{v,R}) \ge \alpha k\}\bigr|
\Bigr]
\ge \gamma n,
\]
Equivalently, averaged over the internal randomness of the algorithm,
linearly many output coordinates cannot be computed by reading $o(k)$ message bits.
\end{corollary}

\begin{proof}
The inequality $\operatorname{Inf}(f) \le D(f)$ is classical for Boolean functions; see, for
example, Rossman's discussion of formula average sensitivity~\cite{Rossman18}.
Therefore every coordinate counted in
\Cref{cor:sat-many-influential-coordinates} also has decision-tree depth $\Omega(k)$.
\end{proof}

What is new here is an explicit satisfiable bounded-occurrence family for which, even allowing
internal randomness and approximation in expectation, linearly many seedwise output bits have
linear deterministic query complexity with respect to the message coordinates, on average over the
seed. This contrasts with local-access results that achieve polylogarithmic query complexity under
strong LLL-type sparsity assumptions or in the high-$k$/low-degree regime~\cite{RTVX11,DongMani24}.

\subsection{Constant-Depth Circuits and Formulas}

We now record a bounded-depth consequence. Boppana's average-sensitivity bound for
constant-depth circuits and Rossman's strengthening for bounded-depth formulas place
polynomial-size bounded-depth computation in the low-influence regime on the message
cube~\cite{Boppana97,Rossman18}. Since our hard family is embedded into the standard SAT
encoding by a linear-size $NC^0$ map, the same barrier transfers to randomized bounded-depth
computation for finding near-satisfying assignments on satisfiable Max E3SAT instances.

\begin{corollary}[No polynomial-size constant-depth randomized circuits for finding near-satisfying assignments]
\label{cor:sat-witness-circuit-lb}
Fix the standard binary encoding of $3$-CNF formulas with $n$ variables and $M$ clauses.
Then there is no randomized polynomial-size constant-depth circuit family that,
on every satisfiable input formula of this size, outputs in expectation a $(1-\eta)$-approximate
assignment. The same holds for randomized bounded-depth formula families.

More quantitatively, for every fixed depth $D$, any such depth-$D$ randomized circuit
family must have size $\exp\bigl(n^{\Omega_D(1)}\bigr)$.
\end{corollary}

\begin{proof}
Suppose there were such a randomized polynomial-size constant-depth circuit family for the
encoded $3$-CNF inputs. Write it as a deterministic circuit family
\[
C_{n,M}(x,r)
\]
with an explicit random seed $r$, and let $A$ denote the induced output distribution when $r$ is
sampled uniformly.
Composing with the hard embedding
\[
m \longmapsto x(m)=\operatorname{Enc}_{n,M}(\Phi(m))
\]
yields, for each fixed seed $r$, a deterministic polynomial-size constant-depth multi-output circuit
\[
F_r(m)=\bigl(f_{1,r}(m),\dots,f_{n,r}(m)\bigr)
\]
on the message cube, because $x$ is computable by a linear-size $NC^0$ circuit.
(The same substitution argument applies to formulas.)
By \Cref{cor:sat-deterministic-cumulative-influence}, one has
\[
\mathbb{E}_{R}
\Bigl[
\sum_{v=1}^n \operatorname{Inf}(f_{v,R})
\Bigr]
= \Omega(nk).
\]

On the other hand, Boppana's average-sensitivity upper bound for bounded-depth circuits
and Rossman's strengthening for bounded-depth formulas imply that every output bit of a
polynomial-size constant-depth circuit or formula has sublinear influence on $\{0,1\}^k$,
uniformly over the fixed seed $r$~\cite{Boppana97,Rossman18}.
Hence
\[
\mathbb{E}_{R}
\Bigl[
\sum_{v=1}^n \operatorname{Inf}(f_{v,R})
\Bigr]
= o(nk),
\]
contradicting \Cref{cor:sat-deterministic-cumulative-influence}
for all sufficiently large $n$.
The quantitative lower bound follows from the quantitative form of the same
average-sensitivity upper bounds.
\end{proof}

The point of \Cref{cor:sat-witness-circuit-lb} is that the obstruction already appears for
randomized generation of near-satisfying assignments on satisfiable SAT formulas. Even though each
satisfiable instance has many acceptable nearly optimal assignments and the algorithm may use internal
randomness, the seedwise output bits still accumulate linear total influence on the hard
family, and this rules out the usual polynomial-size bounded-depth circuit and formula classes.

\bibliographystyle{abbrv}
\bibliography{main}
\appendix
\section[Proof of Lemma~\ref*{lem:lin-sys}]{Proof of \Cref{lem:lin-sys}}\label{sec:lin-sys}

It suffices to exhibit an encoder $f$ whose image is $C$ and whose message coordinates appear in fixed positions after a permutation of coordinates.
Since $C$ is linear, $C = M\mathbb F_p^k$ for some $M\in \mathbb F_p^{n\times k}$ of rank $k$.
By the definition of rank, there exists a set of indices $I=\{i_1, i_2,\dots,i_k\}\subseteq [n]$ such that the rows of $M$ indexed by $I$ span $\mathbb F_p^k$.
Let $\Pi\in\mathbb F_p^{n\times n}$ be the permutation matrix that moves the coordinates in $I$ to the first $k$ positions and let $\widetilde M := \Pi M$.
The first $k$ rows of $\widetilde M$ span $\mathbb F_p^k$, so the projection matrix $P\in\mathbb F_p^{k \times n}$ onto those rows satisfies that $P\widetilde M$ is invertible.
Define $\tilde f(m) := \widetilde M(P\widetilde M)^{-1} m$ for $m\in\mathbb F_p^k$.
We have 
\begin{enumerate}
    \item $\tilde f(\mathbb F_p^k) = \widetilde M\mathbb F_p^k = \Pi C$;
    \item The first $k$ coordinates of $\tilde f(m)$ equal $P\tilde f(m) = P\widetilde M(P\widetilde M)^{-1}m = m$.
\end{enumerate}
Thus $\tilde f$ is a systematic encoder for the permuted code $\Pi C$, witnessing that $\Pi C$ is systematic with the same parameters.
Undoing the permutation shows that $C$ becomes systematic after the fixed coordinate permutation $\Pi$, as claimed.

\end{document}